\documentclass[a4paper,USenglish,cleveref,autoref,thm-restate]{lipics-v2021}

\usepackage[utf8]{inputenc}
\usepackage[T1]{fontenc}
\usepackage{amsmath,amsfonts,amsthm,amssymb}
\usepackage[roman,basic]{complexity}

\usepackage{footnote}
\makesavenoteenv{tabular}
\makesavenoteenv{table}

\usepackage{adjustbox}

\newcommand*{\Otilde}{\widetilde{O}}

\newcommand*{\Offset}{\mathrm{Offset}}

\newcommand*{\last}{\mathrm{pred}}       
\newcommand*{\Stub}{\mathrm{Stub}}
\newcommand*{\Rep}{\mathcal{R}}

\newcommand*{\Last}{\mathrm{Last}}      

\newcommand*{\nwspace}{\hspace*{.1em}} 

\crefname{case}{case}{cases}
\creflabelformat{stat}{#2{\upshape#1}#3}

\crefname{ineq}{inequality}{inequalities}
\creflabelformat{ineq}{#2{(\upshape#1)}#3}

\DeclareMathOperator*{\argmax}{arg\,max}

\let\oldsqrt\sqrt
\def\hksqrt{\mathpalette\DHLhksqrt}
\def\DHLhksqrt#1#2{\setbox0=\hbox{$#1\oldsqrt{#2\,}$}\dimen0=\ht0
   \advance\dimen0-0.2\ht0
   \setbox2=\hbox{\vrule height\ht0 depth -\dimen0}%
   {\box0\lower0.4pt\box2}}
\renewcommand\sqrt\hksqrt

\renewcommand{\leq}{\leqslant}
\renewcommand{\geq}{\geqslant}
\renewcommand{\le}{\leqslant}
\renewcommand{\ge}{\geqslant}
\renewcommand{\epsilon}{\varepsilon}

\usepackage[dvipsnames]{xcolor}

\providecommand{\ignore}[1]{}

\title{Near-Optimal Deterministic Single-Source Distance Sensitivity Oracles}
\titlerunning{Near-Optimal Deterministic Single-Source DSO}

\author{Davide Bilò}%
{Department of Humanities and Social Sciences, University of Sassari, Italy}%
{davidebilo@uniss.it}%
{0000-0003-3169-4300}%
{This work was partially supported by the Research Grant FBS2016\_BILO, funded by ``Fondazione di Sardegna'' in 2016.}

\author{Sarel Cohen}%
{Hasso Plattner Institute, University of Potsdam, Germany}%
{sarel.cohen@hpi.de}%
{} 
{} 

\author{Tobias Friedrich}%
{Hasso Plattner Institute, University of Potsdam, Germany}%
{tobias.friedrich@hpi.de}%
{0000-0003-0076-6308} 
{} 

\author{Martin Schirneck}%
{Hasso Plattner Institute, University of Potsdam, Germany}%
{martin.schirneck@hpi.de}%
{} 
{} 

\authorrunning{D.~Bilò, S.~Cohen, T.~Friedrich, and M.~Schirneck}
\Copyright{Davide Bilò, Sarel Cohen, Tobias Friedrich, and Martin Schirneck}



\nolinenumbers 

\hideLIPIcs  

\begin{document}

\maketitle

\begin{abstract}
Given a graph with a distinguished source vertex $s$,
the Single Source Replacement Paths (SSRP) problem is to compute and output, 
for any target vertex $t$ and edge $e$,
the length $d(s,t,e)$ of a shortest path from $s$ to $t$ that avoids a failing edge $e$.
A Single-Source Distance Sensitivity Oracle (Single-Source DSO)
is a compact data structure that answers queries of the form $(t,e)$
by returning the distance $d(s,t,e)$.
We show how to compress the output of the SSRP problem 
on $n$-vertex, $m$-edge graphs with integer edge weights in the range $[1,M]$
into a deterministic Single-Source DSO that has size $O(M^{1/2} n^{3/2})$ and query time $\Otilde(1)$.
We prove that the space requirement is optimal (up to the word size).
Our techniques can also handle vertex failures within the same bounds.

Chechik and Cohen [SODA 2019] presented a combinatorial, randomized $\Otilde(m\sqrt{n}+n^2)$
time SSRP algorithm for undirected and unweighted graphs.
We derandomize their algorithm with the same asymptotic running time and apply our compression to obtain a deterministic Single-Source DSO with $\Otilde(m\sqrt{n}+n^2)$ preprocessing time, $O(n^{3/2})$ space, and $\Otilde(1)$ query time. Our combinatorial Single-Source DSO has near-optimal
space, preprocessing and query time for unweighted graphs,
improving the preprocessing time by a $\sqrt{n}$-factor 
compared to previous results with $o(n^2)$ space.

Grandoni and Vassilevska Williams [FOCS 2012, TALG 2020] gave an algebraic, randomized $\Otilde(Mn^\omega)$ time SSRP algorithm for (undirected and directed) graphs with integer edge weights in the range $[1,M]$, 
where $\omega < 2.373$ is the matrix multiplication exponent.
We derandomize it for undirected graphs and apply our compression to obtain an algebraic Single-Source DSO with $\Otilde(Mn^\omega)$ preprocessing time, $O(M^{1/2} \nwspace n^{3/2})$ space,
and $\Otilde(1)$ query time.
This improves the preprocessing time of algebraic Single-Source DSOs
by polynomial factors compared to previous $o(n^2)$-space oracles.

We also present further improvements of our Single-Source DSOs. 
We show that the query time  can be reduced to a constant 
at the cost of increasing the size of the oracle to $O(M^{1/3} \nwspace n^{5/3})$
and that all our oracles can be made path-reporting.
On sparse graphs with $m=O\big(\frac{n^{5/4-\varepsilon}}{M^{7/4}}\big)$ edges,
for any constant $\varepsilon > 0$, we reduce the preprocessing to randomized $\Otilde(M^{7/8} \nwspace m^{1/2} \nwspace n^{11/8}) = O(n^{2-\varepsilon/2})$ time.
To the best of our knowledge, this is the first truly subquadratic time algorithm for building Single-Source DSOs on sparse graphs. 
\end{abstract}

\ignore{%
Given a graph with a distinguished source vertex $s$, the Single Source Replacement Paths (SSRP) problem is to compute and output, for any target vertex $t$ and edge $e$, the length $d(s,t,e)$ of a shortest path from $s$ to $t$ that avoids a failing edge $e$.
A Single-Source Distance Sensitivity Oracle (Single-Source DSO) is a compact data structure that answers queries of the form $(t,e)$ by returning the distance $d(s,t,e)$.
We show how to compress the output of the SSRP problem  on $n$-vertex, $m$-edge graphs with integer edge weights in the range $[1,M]$
into a deterministic Single-Source DSO that has size $O(M^{1/2} n^{3/2})$ and query time $\widetilde{O}(1)$.
We prove that the space requirement is optimal (up to the word size).
Our techniques can also handle vertex failures within the same bounds.

Chechik and Cohen [SODA 2019] presented a combinatorial, randomized $\widetilde{O}(m\sqrt{n}+n^2)$ time SSRP algorithm for undirected and unweighted graphs.
We derandomize their algorithm with the same asymptotic running time and apply our compression to obtain a deterministic Single-Source DSO with $\widetilde{O}(m\sqrt{n}+n^2)$ preprocessing time, $O(n^{3/2})$ space, and $\widetilde{O}(1)$ query time. 
Our combinatorial Single-Source DSO has near-optimal space, preprocessing and query time for unweighted graphs, improving the preprocessing time by a $\sqrt{n}$-factor compared to previous results with $o(n^2)$ space.

Grandoni and Vassilevska Williams [FOCS 2012, TALG 2020] gave an algebraic, randomized $\widetilde{O}(Mn^\omega)$ time SSRP algorithm for (undirected and directed) graphs with integer edge weights in the range $[1,M]$, where $\omega < 2.373$ is the matrix multiplication exponent.
We derandomize it for undirected graphs and apply our compression to obtain an algebraic Single-Source DSO with $\widetilde{O}(Mn^\omega)$ preprocessing time, $O(M^{1/2} \nwspace n^{3/2})$ space, and $\widetilde{O}(1)$ query time.
This improves the preprocessing time of algebraic Single-Source DSOs by polynomial factors compared to previous $o(n^2)$-space oracles.

We also present further improvements of our Single-Source DSOs. 
We show that the query time  can be reduced to a constant at the cost of increasing the size of the oracle to $O(M^{1/3} n^{5/3})$ and that all our oracles can be made path-reporting.
On sparse graphs with $m=O(\frac{n^{5/4-\varepsilon}}{M^{7/4}})$ edges, for any constant $\varepsilon > 0$, we reduce the preprocessing to randomized $\widetilde{O}(M^{7/8} m^{1/2} n^{11/8}) = O(n^{2-\varepsilon/2})$ time.
To the best of our knowledge, this is the first truly subquadratic time algorithm for building Single-Source DSOs on sparse graphs.
}


\section{Introduction}
\label{sec:intro}

One of the basic problems in computer science is the computation of shortest paths and distances in graphs that are subject to a small number of transient failures.
We study two central 
problems of this research area on undirected graphs $G$ with $n$ vertices and $m$ edges, namely, the \emph{Single-Source Replacement Paths} (SSRP) problem and \emph{Single-Source Distance Sensitivity Oracles} (Single-Source DSOs).

\textbf{\textsf{The SSRP Problem.}}
In the SSRP problem, we are given a graph $G$ with a fixed source vertex $s$
and are asked to compute, for every vertex $t$ and edge $e$,
the replacement distance $d(s,t,e)$,
which is the length of the shortest $s$-$t$-path in the graph $G\,{-}\,e$, 
obtained by dropping the edge $e$.
By first computing any shortest path tree for $G$ rooted at $s$,
one can see that there are only $O(n^2)$ relevant distances $d(s,t,e)$,
namely, those for which $e$ is in the tree.

Chechik and Cohen~\cite{ChechikCohen19SSRP_SODA} presented an $\Otilde(m\sqrt{n} + n^2)$ time\footnote{%
	For a non-negative function $f = f(n)$,
	we use $\Otilde(f)$ to denote $O(f \cdot \textsf{polylog}(n))$.
} 
combinatorial\footnote{%
    The term ``combinatorial algorithm'' is not well-defined, and is often interpreted as not using any matrix multiplication.
    Arguably, combinatorial algorithms can be considered 
    efficient in practice as the constants hidden in the matrix
    multiplication bounds are rather high.
} 
SSRP algorithm for unweighted graphs.\
They also showed that the running time cannot be improved by polynomial factors,
assuming that any combinatorial algorithm for
Boolean Matrix Multiplication (BMM)
on $n \times n$ matrices containing $m$ 1's requires $mn^{1-o(1)}$ time. 
Gupta~et~al.~\cite{GJM20} simplified the SSRP algorithm and generalized it to multiple sources.
For a set of $\sigma$ sources, they presented a combinatorial algorithm that takes
$\Otilde(m\sqrt{n\sigma} + \sigma n^2)$ time.
Grandoni and Vassilevska Williams~\cite{GW12,GrandoniVWilliamsFasterRPandDSO_journal} gave an algorithm for both directed and undirected graphs with integer edge weights in the range $[1,M]$ that uses fast matrix multiplications and
runs in $\Otilde(Mn^\omega)$ time, where $\omega<2.37286$ is the matrix multiplication exponent~\cite{AlmanVWilliams21RefinedLaserMethod,LeGall14,williams2012multiplying}.
We are only concerned with \emph{positive} integer weights,
but it is worth noting that SSRP with weights in $[-M,M]$
is strictly harder, modulo a breakthrough in Min-Plus Product computation,
with a current best running time of $O(M^{0.8043} \nwspace n^{2.4957})$
as shown by Gu et al.~\cite{GuPolakVWilliamsXu21MonotoneMinPlus}.

All the SSRP algorithms above are randomized, it is an interesting open problem 
whether they can be derandomized in the same asymptotic running time.

\textbf{\textsf{Single-Source DSOs.}}
A Distance Sensitivity Oracle (DSO)
is a data structure that answers queries $(u,v,e)$, for vertices $u, v$ and edge $e$,
by returning the replacement distance $d(u,v,e)$, 
A Single-Source DSO, with fixed source $s$, answers queries $(t,e)$ with $d(s,t,e)$.

Of course, any SSRP algorithm gives a Single-Source DSO 
by just tabulating the whole output in $O(n^2)$ space,
the replacement distances can then be queried in constant time.
However, the space usage is far from optimal.
Parter and Peleg~\cite{PaPeP16} developed a deterministic algorithm that computes an $O(n^{3/2})$ size subgraph of $G$ containing a breadth-first-search tree of $G\,{-}\,e$ for every failing edge $e$. The subgraph can also be thought of as a Single-Source DSO with $O(n^{3/2})$ space and query time.
Bilò~et~al.~\cite{BCGLPP18} presented a Single-Source DSO of the same size with $\Otilde(\sqrt{n})$ query time and $\Otilde(mn)$ preprocessing time. 
Gupta and Singh~\cite{GuptaSingh18FaultTolerantExactDistanceOracle} later designed a randomized Single-Source DSO of $\Otilde(n^{3/2})$ size, $\Otilde(mn)$ preprocessing time,\footnote{
	The authors of~\cite{GuptaSingh18FaultTolerantExactDistanceOracle} do not report the
	preprocessing time, but it can be reconstructed as $\Otilde(mn)$.
}
but with a better $\Otilde(1)$ query time. The results in the latter two works generalize to the case of $\sigma$ sources in such a way that the time and size scale by $o(\sigma)$ factors.  

For the case of $\sigma=n$ sources, that is, general (all-pairs) DSOs, 
Bernstein and Karger~\cite{BeKa08,BeKa09} designed an oracle taking
$\Otilde(n^2)$ space with constant query time, even for directed graphs with real edge weights.
The space was subsequently improved to $O(n^2)$ by Duan and Zhang~\cite{DuanZhang17ImprovedDSOs},
which is optimal~\cite{ThorupZ05}.
The combinatorial $\Otilde(mn)$ time preprocessing for building the DSOs is conditionally
near-optimal as it matches the best known bound (up to polylogarithmic factors) 
for the simpler problem of finding the All-Pairs Shortest Paths (APSP).
The conditional lower bound in~\cite{ChechikCohen19SSRP_SODA}, stating that there exists no combinatorial algorithm solving the undirected SSRP problem with real edge weights in $O(mn^{1-\varepsilon})$ time for any positive $\varepsilon > 0$, unless there is a combinatorial algorithm for the APSP problem in
$O(mn^{1-\varepsilon})$ time, also implies that there exists no \emph{Single-Source} DSO with $\Otilde(1)$ query time and $O(mn^{1-\varepsilon})$ preprocessing time for real edge weights. Therefore, the DSOs in~\cite{BeKa09,DuanZhang17ImprovedDSOs},
are also conditionally near-optimal for the single source case with real edge weights. 

Several algebraic all-pairs DSOs with subcubic preprocessing time
have been developed in the last decade
for graphs with integer edge weights in $[1,M]$~\cite{BrSa19,ChCo20,GW12,Ren20,WY13}.
Very recently, Gu and Ren~\cite{GuRen21ConstructingDSO_ICALP} presented a randomized DSO 
achieving a $O(M n^{2.5794})$ preprocessing time with $O(1)$ query time,
improving upon the one by Ren~\cite{Ren20,Ren20_arxiv}
with an $\Otilde(M n^{2.6865})$ preprocessing time.
Those DSOs can also be used in the single-source case,
but the requirement to store the information for all pairs 
forces them to take $\Omega(n^2)$ space~\cite{ThorupZ05}.
The algebraic SSRP algorithm in~\cite{GrandoniVWilliamsFasterRPandDSO_journal},
seen as a data structure, has a better preprocessing time than any known (general) DSO
but also takes $O(n^2)$ space, which we have seen to be wasteful.

We are not aware of an algebraic Single-Source DSO that simultaneously
achieves $o(n^2)$ space and has a better preprocessing time than their all-pairs counterparts.
It is interesting whether we can construct space-efficient oracles
faster when focusing on a single source. 

Additional results on replacement paths and DSOs (for single or multiple failures and directed graphs) can be found in \cite{AlonChechikCohen19CombinatorialRP, BrSa19, CLPR10,ChMa19, DeThChRa08, DP09, GrandoniVWilliamsFasterRPandDSO_journal, HeSu01, HeSu01_erratum, MaMiGu89, NaPrWi01,Nardelli2003, RodittyZwick12kSimpleShortestPaths, WilliamsW10_journal}.
The most efficient Single-Source DSOs in their respective settings are shown in \autoref{tab:comparison} below.

\subsection{Our Contribution}

We research SSRP algorithms, Single-Source DSO data structures,
and the connection between the two.
Our first contribution is presented in \autoref{sec:derandomize}.
We derandomize the near-optimal combinatorial SSRP algorithm of
Chechik and Cohen~\cite{ChechikCohen19SSRP_SODA} for undirected, unweighted graphs
and the algebraic algorithm of 
Grandoni and Vassilevska Williams~\cite{GW12,GrandoniVWilliamsFasterRPandDSO_journal}
for undirected graphs with integer weights in the range $[1,M]$.
Both deterministic algorithms have the same asymptotic runtime as their randomized counterparts.

\begin{theorem}
\label{thm:ssrp-deterministic}
	There is a deterministic, combinatorial SSRP algorithm for undirected, unweighted graphs running in
	time $\Otilde(m\sqrt{n}+n^2)$
	and a deterministic, algebraic SSRP algorithm for undirected graphs
	with integer weights in the range $[1,M]$ running in $\Otilde(Mn^\omega)$. 
\end{theorem}

We present in \autoref{sec:reduction_algorithm_to_DS} a deterministic reduction
from the problem of building a Single-Source DSO to SSRP 
on undirected graphs with small integer edge weights.

\begin{theorem}
\label{thm:reduction}
	Let $G$ be an undirected graph with integer edge weights in the range $[1,M]$
	and let $s$ be the source vertex.
	Suppose we are given access to a shortest path tree $T_s$ of $G$
	rooted in $s$ and all values $d(s,t,e)$ for vertices $t$ of $G$ and edges $e$ in $T_s$.
	There is a deterministic, combinatorial algorithm that in time $O(n^2)$
	builds a Single-Source DSO of size $O(M^{1/2} \nwspace n^{3/2})$ with $\Otilde(1)$ query time.
	The same statement holds for vertex failures if instead we are given access to 
	the values $d(s,t,v)$ for all vertices $t$ and $v$ of $G$ 
\end{theorem}

\noindent
The algorithm does not require access to the graph $G$ itself.
As there can be up to $O(n^2)$ relevant distances $d(s,t,e)$,
the running time is linear in the input.
If the algorithm additionally has access to $G$ and is given $O(m\sqrt{Mn}+n^2)$ time,
the Single-Source DSO 
also reports the replacement paths $P(s,t,e)$ in time $\Otilde(1)$ per edge.
The query time of the oracle can be improved to $O(1)$
at the cost of increasing the size of the oracle to $O(M^{1/3} \nwspace n^{5/3})$. 

Plugging the deterministic SSRP algorithms of \autoref{thm:ssrp-deterministic} 
into our reduction of \autoref{thm:reduction},
gives the following Single-Source DSOs as corollaries. 

\begin{theorem}
\label{thm:SSDSO}
	There is a deterministic, combinatorial Single-Source DSO for undirected, unweighted graphs
	taking $O(n^{3/2})$ space, with $\Otilde(m\sqrt{n}+n^2)$ preprocessing time, 
	and $\Otilde(1)$ query time.
	There is a deterministic, algebraic Single-Source DSO for undirected graphs
	with integer weights in the range $[1,M]$ taking $O(M^{1/2} \nwspace n^{3/2})$ space,
	with $\Otilde(Mn^\omega)$ preprocessing time, and $\Otilde(1)$ query time.
\end{theorem}

When comparing the results with other Single-Source DSO with $o(n^2)$ space,
the preprocessing time of our combinatorial solution is better by a factor of $\sqrt{n}$
compared to previous oracles~\cite{BCGLPP18, GuptaSingh18FaultTolerantExactDistanceOracle}.
The preprocessing time of the algebraic part of \autoref{thm:SSDSO}
improves (ignoring polylogarithmic factors) by a factor of 
$n^{2.5794-\omega} > n^{0.2}$
over the current best algebraic (all-pairs) DSO~\cite{GuRen21ConstructingDSO_ICALP}.
See \autoref{tab:comparison} for more details.
In fact, we combine the efficient preprocessing of SSRP algorithms (seen as DSOs)
with a compression scheme that achieves nearly-optimal space.
To the best of our knowledge, \autoref{thm:SSDSO} presents the first algebraic Single-Source DSO
with $o(n^2)$ space
that achieves a better performance than any all-pairs DSO.
It is also the first space-efficient Single-Source DSO for graphs with small integer weights.

We further study lower bounds for Single-Source DSOs. 
Note that given an oracle whose preprocessing time is $P$ and query time is $Q$, 
one can solve the SSRP problem in time $P+n^2 \cdot Q$ by building the DSO 
and running the queries $(t,e)$ for every $t \in V, e \in E(T_s)$.
Therefore, if $n^2 \cdot Q = O(P)$, the $mn^{1/2-o(1)}+\Omega(n^2)$
conditional\footnote{
	The $\Omega(n^2)$ term is unconditional
	and stems from the size of the output, see~\cite{ChechikCohen19SSRP_SODA}.
}
time-lower bound for the SSRP problem~\cite{ChechikCohen19SSRP_SODA},
obtained by a reduction from BMM, 
implies the same lower bound for $P$.
The preprocessing of our combinatorial oracle in \autoref{thm:SSDSO}
is thus nearly optimal.
We further investigate how much a Single-Source DSO can be compressed.
In contrast to~\cite{ChechikCohen19SSRP_SODA},
we obtain an \emph{unconditional} \emph{space}-lower bound using an argument from information theory.

\begin{restatable}{theorem}{spacelowerbound}
\label{thm:space_lower_bound}
	Any Single-Source DSO must take $\Omega(\min\{M^{1/2} \nwspace n^{3/2}, \nwspace n^2\})$ 
	bits of space on at least one $O(n)$-vertex graph with integer edge weights in the range $[1,M]$.
\end{restatable}

\noindent
A small gap remains between \Cref{thm:space_lower_bound,thm:reduction}
as the space is bounded at $\Omega( M^{1/2} \nwspace n^{3/2})$ bits,
while the oracle takes this many machine words.
Nevertheless, it shows that on dense graphs our Single-Source DSOs in \autoref{thm:SSDSO} have near-optimal space. 

The Single-Source DSOs presented above all have $\Omega(n^2)$ preprocessing time,
which cannot be avoided for graphs with $m=\Omega(n^{3/2})$, assuming the BMM conjecture.
SSRP algorithms require $\Omega(n^2)$ time simply to output the solution.
It is not clear whether this lower bound also applies to Single-Source DSO on sparse graphs.
We partially answer this question negatively by developing a truly subquadratic, randomized Single-Source DSO in \autoref{sec:subquadratic_preprocessing}.
We use new algorithmic techniques and structural properties of independent interest.

\begin{theorem}
\label{thm:subquadratic_preprocessing}
	There is a randomized Single-Source DSO taking $O(M^{1/2} \nwspace n^{3/2})$ space
	that has $\Otilde(1)$ query time w.h.p.\footnote{%
		An event occurs \emph{with high probability} (w.h.p.)
		if it has probability at least $1 - n^{-c}$ for some $c > 0$.
	}
	The oracle also reports a replacement path in $\Otilde(1)$ time per edge w.h.p. 
	On graphs with $m=O(M^{3/4} \nwspace n^{7/4})$ edges,
	the preprocessing time is
	$\Otilde(M^{7/8} \nwspace m^{1/2} \nwspace n^{11/8})$.
	If the graph is sparse, 
	meaning $m = O(n^{5/4-\varepsilon}/M^{7/4})$ for any $\varepsilon > 0$,
	this is $\Otilde(n^{2-{\varepsilon/2}})$.
\end{theorem}

\subsection{Comparison with Previous Work}
\label{subsec:intro_previous_work}

\begin{table}
 \centering
 \renewcommand{\arraystretch}{1.1}
 \begin{tabular}{lllc@{\ }c@{\ }c@{\ }cl}
 	\noalign{\hrule height 1pt}
 	\bf{Preprocessing time} &
 	\bf{Space} &
 	\bf{Query time} &
 	\multicolumn{4}{c}{\textbf{Setting}} &
 	\bf{Reference} 
 	\\
 	\noalign{\hrule height 1pt}\\[-12pt]
 	$\Otilde(m n)$ &
 	$O(n^2)$ &
 	$O(1)$ &
 	D &
 	C &
 	W &
 	Ap &
 	\cite{BeKa09,DuanZhang17ImprovedDSOs}
 	\\
 	$\Otilde(M n^{2.5794})$ &
 	$\Otilde(n^{2})$ &
 	$O(1)$ &
 	R &
 	A &
 	I &
 	Ap &
 	\cite{GuRen21ConstructingDSO_ICALP}
    \\[.125em]
 	\noalign{\hrule height 1pt}\\[-12pt]
 	$\Otilde(mn^{1/2} + n^2)$ &
 	$O(n^{2})$ &
 	$O(1)$ &
 	R &
 	C &
 	U &
 	Ss &
 	\cite{ChechikCohen19SSRP_SODA}
 	\\
 	$\Otilde(M n^{\omega})$ &
 	$O(n^{2})$ &
 	$O(1)$ &
 	R &
 	A &
 	I &
 	Ss &
 	\cite{GrandoniVWilliamsFasterRPandDSO_journal}
 	\\
 	$\Otilde(m n)$ &
 	$\Otilde(n^{3/2})$ &
 	$\Otilde(n^{1/2})$ &
 	D &
 	C &
 	U &
 	Ss &
 	\cite{BCGLPP18}
 	\\
 	$\Otilde(m n)$ &
 	$\Otilde(n^{3/2})$ &
 	$\Otilde(1)$ &
 	R &
 	C &
 	U &
 	Ss &
 	\cite{GuptaSingh18FaultTolerantExactDistanceOracle}
 	\\[.125em]
 	\noalign{\hrule height 1pt}\\[-12pt]
 	$\Otilde(m n^{1/2} + n^2)$ &
 	$O(n^{5/3})$ &
 	$O(1)$ &
 	D &
 	C &
 	U &
 	Ss &
    \autoref{lem:constant_query_time}
    \\
 	$\Otilde(m n^{1/2}+n^2)$ &
 	$O(n^{3/2})$ &
 	$\Otilde(1)$ &
 	D &
 	C &
 	U &
 	Ss &
 	\autoref{thm:SSDSO}
    \\
     $\Otilde(M n^{\omega})$ &
 	$O(M^{1/3} \nwspace n^{5/3})$ &
 	$O(1)$ &
 	D &
 	A &
 	I &
 	Ss &
 	\autoref{lem:constant_query_time}
    \\
 	$\Otilde(M n^{\omega})$ &
 	$O(M^{1/2} \nwspace n^{3/2})$ &
 	$\Otilde(1)$ &
 	D &
 	A &
 	I &
 	Ss &
 	\autoref{thm:SSDSO}
    \\
 	$\Otilde(M^{7/8} \nwspace m^{1/2} \nwspace n^{11/8})$\textsuperscript{\textdagger} &
 	$O(M^{1/2} n^{3/2})$ &
 	$\Otilde(1)$ &
 	R &
 	C &
 	I &
 	Ss &
 	\autoref{thm:subquadratic_preprocessing} 
  	\\
 	\noalign{\hrule height 1pt}
 	\vspace{.25em}
 \end{tabular}
 \caption{
 	Comparison of results. 
 	\textsuperscript{\textdagger}The preprocessing time is for graphs with $m = O(M^{3/4} n^{7/4})$.
 }
 \label{tab:comparison}
 \end{table}

\noindent
 \autoref{tab:comparison} shows a comparison of the most efficient Distance Sensitivity Oracles in their respective setting, as well as the results presented in this work.
 We distinguish four dimensions of different problem types.
 \begin{enumerate}
 \item
 Randomized (R) vs.\ deterministic (D),
 \item
 Combinatorial (C) vs.\ algebraic (A),
 \item
 Unweighted (U) vs.\ real weights (W) vs.\ integer weights in $[1,M]$ (I),
 \item
 All-Pairs (Ap) vs.\ single-source (Ss).
 \end{enumerate}
 
 \noindent
 Our deterministic, combinatorial Single-Source DSO from \autoref{thm:SSDSO} has near-optimal space, preprocessing and query time for dense graphs.
 It improves the preprocessing time of the randomized DSOs by Bernstein and Karger~\cite{BeKa09}, Bilò et al.~\cite{BCGLPP18}, and Gupta and Singh~\cite{GuptaSingh18FaultTolerantExactDistanceOracle} by a factor of $O(\sqrt{n})$.
 When viewing the randomized SSRP algorithm of Chechik and Cohen as an oracle,
 our solution has the same preprocessing time but reduces the space requirement,
 by an near-optimal factor of $O(n^{1/3})$ while increasing the query time to only $\Otilde(1)$.
 
 Our algebraic combinatorial Single-Source DSO from \autoref{thm:SSDSO} has near-optimal space and query time for dense graphs, its preprocessing time improves over the randomized, algebraic DSOs of Chechick and Cohen~\cite{ChCo20}, Ren~\cite{Ren20,Ren20_arxiv},
as well as Gu and Ren~\cite{GuRen21ConstructingDSO_ICALP} by a factor of $\Otilde(n^{2.5794-\omega})$.
 It has the same preprocessing time as the SSRP algorithm 
 by Grandoni and Vassilevska Williams~\cite{GrandoniVWilliamsFasterRPandDSO_journal},
 but compresses the output to $O(M^{1/2} \nwspace n^{3/2})$ space.

Our Single-Source DSO from \autoref{lem:constant_query_time} even achieves constant query time
 at the expense of larger $O(n^{5/3})$ (respectively, $O(M^{1/3} \nwspace n^{5/3})$) space. 
All of our oracles can handle vertex failures and are path-reporting,
the query time then corresponds to the time needed per edge of the replacement path. 
 In \autoref{thm:subquadratic_preprocessing}, we also obtain Single-Source DSO with subquadratic preprocessing time for sparse graphs.

\subsection{Techniques}
\label{subsec:intro_techniques}

\textbf{\textsf{Multi-stage derandomization.}}
To derandomize the SSRP algorithms, we extend the techniques by Alon, Chechik, and Cohen~\cite{AlonChechikCohen19CombinatorialRP} to identify a small set of critical paths
we need to hit.
In~\cite{AlonChechikCohen19CombinatorialRP}, a single set of paths was sufficient,
we extend this to a hierarchical multi-stage framework.
The set of paths in each stage depends on the hitting set found in the previous ones.
For example, a replacement path from $s$ to $t$ avoiding the edge $e$ decomposes into
two shortest paths $P(s,q)$ and $P(q,t)$ in the original graph for some unknown vertex $q$, see~\cite{Afek02RestorationbyPathConcatenation_journal}.
It is straightforward to hit all of the components $P(s,q)$.
We then use this hitting set in a more involved way to find sets of vertices
that also intersect all of the subpaths $P(q,t)$.

\textbf{\textsf{Versatile compression.}}
The key observation of our reduction to SSRP is that any shortest $s$-$t$-path
can be partitioned into $O(\sqrt{Mn})$ segments
such that all edges in a segment have the same replacement distance. 
Gupta and Singh~\cite{GuptaSingh18FaultTolerantExactDistanceOracle} proved this for unweighted graphs.
However, it is not obvious how to generalize their approach to the weighted case.
We give a simpler proof in the presence of small integer weights,
which immediately transfers also to vertex failures.
We further show how to extend this to multiple targets
and even reuse it to obtain the subquadratic algorithm on sparse graphs. 
In~\cite{GuptaSingh18FaultTolerantExactDistanceOracle},
a randomized oracle was presented that internally uses the
rather complicated data structures of Demetrescu et al.~\cite{DeThChRa08}.
We instead give a deterministic construction implementable with only a few arrays.
Unfortunately, the compression scheme crucially depends on the graph being undirected. 

\textbf{\textsf{Advanced search for replacement paths.}}
The randomized algorithm building the DSO in subquadratic time for sparse graphs
needs to find the $O(\sqrt{Mn})$ segments partitioning the $s$-$t$-path. Naively, this takes $O(n)$ time per target vertex $t$ as we need to explore the whole path for potential segment endpoints
and do not know the corresponding replacement paths in advance. We use standard random sampling to hit all such replacements paths with only a few vertices
and exploit the path's monotonicity properties to develop more advanced search techniques. 
This reduces the time needed per target to $O(n^{1-\varepsilon})$, after some preprocessing.
The analysis uses the fact that entire subpaths can be discarded without exploration.

\subparagraph{Open problems.}
Our compression scheme and the randomized, subquadratic Single-Source DSO on sparse graphs
can also handle vertex failures rather than only edge failures. 
It remains an open question whether one can obtain efficient deterministic 
SSRP algorithms in the vertex-failure scenario.
If an analog of \Cref{thm:ssrp-deterministic} held for vertex failures,
then  \autoref{thm:reduction} would directly transfer the extension also to the DSOs of \autoref{thm:SSDSO}.
Another interesting open question is whether there is a Single-Source DSO with deterministic, truly subquadratic time preprocessing on graphs with $m=O(n^{3/2-\varepsilon})$ edges.
Can one obtain better Single-Source DSOs, and prove matching lower bounds, for sparse graphs?

\section{Preliminaries}
\label{sec:prelims}

We let $G=(V,E,w)$ denote the undirected, edge-weighted base graph
on $n$ vertices and $m$ edges,
and tacitly assume $m \ge n$.
The weights $w(e)$, $e \in E$, are integers in $[1,M]$ with $M = \textsf{poly}(n)$.
For an undirected, weighted graph $H$,
we denote by $V(H)$ the set of its vertices, and by $E(H)$ edge set of its edges.
We write $e \in H$ for $e \in E(H)$ and $v \in H$ for $v \in V(H)$. 
Let $P$ be a simple path in $H$.
The \emph{length} or \emph{weight} $w(P)$ of $P$ is $\sum_{e \in E(P)} w(e)$.
For $u,v \in V(H)$, we denote by $P_H(u,v)$ a shortest path (one of minimum weight) 
from $u$ to $v$.
If a particular shortest path is intended, this will be made clear from the context.
The \emph{distance} of $u$ and $v$ is $d_H(u,v) = w(P_H(u,v))$.
We drop the subscript when talking about the base graph $G$.
The restriction on the maximum weight $M$ allows us to store any graph distance
in a single machine word on $O(\log n)$ bits.
Unless explicitly stated otherwise, we measure space complexity in the number of words.

Let $x,y \in V(P)$ be two vertices on the simple path $P$.
We denote by $P[x..y]$ the subpath of $P$ from $x$ to $y$.
Let $P_1 = (u_1, \dots, u_i)$ and $P_2 = (v_1, \dots, v_j)$ be two paths in $H$.
Their \emph{concatenation} is $P_1 \circ P_2 = (u_1, \dots, u_i, v_1, \dots, v_j)$,
provided that $u_i = v_1$ or $\{u_i,v_1\} \in E(H)$.

Fix some \emph{source vertex} $s \in V$ in the base graph $G$.
For any \emph{target vertex} $t \in V$ and edge $e \in E$,
we let $P(s,t,e)$ denote a \emph{replacement path} for $e$,
that is, a shortest path from $s$ to $t$ in $G$ that does not use the edge $e$.
Its weight $d(s,t,e) = w(P(s,t,e))$ is the \emph{replacement distance}.
Given a specific shortest path $P(s,t)$ in $G$ and a replacement path $P(s,t,e)$,
we can assume w.l.o.g.\ that the latter is composed of
the \emph{common prefix} 
that it shares with $P(s,t)$, 
the \emph{detour} part 
which is edge-disjoint from $P(s,t)$,
and the \emph{common suffix} after $P(s,t,e)$ remerges with $P(s,t)$.
%
%
All statements apply to vertex failures as well.

\section{Using SSRP to Build Single-Source DSOs}
\label{sec:reduction_algorithm_to_DS}

In this section, we prove \autoref{thm:reduction}.
We describe how to deterministically reduce the task of building a Single-Source DSO
to computing the replacement distances in the SSRP problem.
Recall that we assume we are given a shortest path tree $T_s$ 
of the base graph $G$ rooted in the source $s$.
This does not loose generality as we could as well compute it in time $O(m)$ 
via Thorup's algorithm~\cite{Thorup99UndirectedSSSPinLinearTime}.
However, the tree $T_s$ focuses our attention to the $O(n^2)$
\emph{relevant} replacement distances in $G$.
The failure of an edge $e$ can only increase the distance from $s$ to some target $t$
if $e$ lies on the $s$-$t$-path $P(s,t)$ in $T_s$.
Given a query $(e,t)$, we can thus check whether $e$ is relevant for $t$ in $O(1)$ time
using a lowest common ancestor (LCA) data structure of size $O(n)$~\cite{BenderFarachColton00LCARevisited}.
If the maximum weight $M$ is larger than $n$, we are done
as we store the relevant replacement distances, original graph distances,
and the LCA data structure.

However, for $M \le n$, there are more space-efficient solutions.
Using time $O(n^2)$, that is, linear in the number of relevant distances,
we compress the space needed to store them down to $O(M^{1/2} \nwspace n^{3/2})$
while increasing the query time only to $\Otilde(1)$.
This scheme also allows several extensions,
namely, handling vertex failures, reporting fault-tolerant shortest path trees, 
or retaining constant query time by using slightly more space.
We first give an overview of the reduction.
Suppose we have a set of \emph{pivots} $D \subseteq V$
such that any $s$-$t$-path $P(s,t)$ in $T_s$ has at least one element of $D$ 
among its last $\sqrt{n}$ vertices.
For a target $t$, let $x$ be the pivot on $P(s,t)$ that is closest to $t$.
We distinguish three cases depending on the failing edge $e$.

\begin{itemize}
	\item\label[case]{case:near} \textbf{\textsf{Near case.}}
		The edge $e$ belongs to the near case if it is
		on the subpath $P(s,t)[x..t]$ from the last pivot to the target. 
		We construct a data structure to quickly identify those edges
		It is then enough to store the associated replacement distances explicitly.
	\item\label[case]{case:far_I} \textbf{\textsf{Far case I.}}
		The edge $e$ belongs to the far case I if it is on the subpath $P(s,t)[s..x]$
		and there is a replacement path for $e$ that uses the vertex $x$.
		We handle this by storing a linear number of distances for every pivot in $D$.
	\item\label[case]{case:far_II} \textbf{\textsf{Far case II.}} 
		We are left with edges $e$ on $P(s,t)[s..x]$ for which no replacement path uses $x$.
		We show that there are only $O(M^{1/2} \nwspace n^{3/2})$
		many different replacement distances of this kind.
		We can find the correct distance in $\Otilde(1)$.
		This is the only case with a quadratic running time,
		space requirements depending on $M$,
		and a super-constant query time.
		We also show how to avoid the latter at the expense of a higher space complexity.
\end{itemize}

\subparagraph*{Near case.}

We first describe how to obtain the set $D$.
We also take $D$ to denote a representing data structure.
That is, for all $t \in V$, $D[t]$ shall denote the last pivot on the path $P(s,t)$ in $T_s$.
A deterministic greedy algorithm efficiently computes a small sets $D$.

\begin{restatable}{lemma}{setofpivots}
\label{lem:set_of_pivots}
	There exists a set $D \subseteq V$ with $|D| \le \sqrt{n}$, computable in time $\Otilde(n)$,
 	such that every $s$-$t$-path in $T_s$
	contains a pivot in $D$ among its last $\sqrt{n}$ vertices.
	In the same time, we can compute a data structure taking $O(n)$ space
	that returns $D[t]$ in constant time.
\end{restatable}

Let $x = D[t]$ be the pivot assigned to $t$.
An edge $e$ belongs to the \emph{near case} with respect to $t$
if it lies on $P(x,t) = P(s,t)[x..t]$.
Observe that $P(x,t)$ has less than $\sqrt{n}$ edges.
We store $d(s,t,e)$ for the near case in an array with $(t,e)$ as key.
With access to the distances, the array can be computed in $O(n^{3/2})$ total time and space.
Consider a query $(t,e)$ such that $e$ has already been determined above to be on the path $P(s,t)$.
The edge $e = \{u,v\}$ thus belongs to the near case iff $D[u] = D[v] = x$.
If so, we look up $d(s,t,e)$ in the array.

\subparagraph*{Far case I.}

We say a query $(t,e)$ belongs to the \emph{far case} 
if $e$ is on the subpath $P(s,x) = P(s,t)[s..x]$.
These are the queries not yet handled by the process above.
Note that $d(s,t,e) \le d(s,x,e) + d(x,t)$ holds for all queries in the far case.
If a replacement path $P(s,x,e)$ exists,
$P(s,x,e) \circ P(s,t)[x..t]$ is some $s$-$t$-path that avoids $e$ whose length is the right-hand side;
otherwise, we have $d(s,x,e) = \infty$ and $d(s,t,e) \le d(s,x,e) + d(x,t)$ holds vacuously.
We split the far case depending on the existence of certain replacement paths.
Recall that we can assume that any replacement path
consists of a common pre- and suffix with the original path $P(s,t)$
and a detour that is edge-disjoint from $P(s,t)$.
We let $(t,e)$ belong to the \emph{far case I} if $e$ is on $P(s,x)$
and there is a replacement path $P(s,t,e)$ that uses the vertex $x$.
It is readily checked that for a query in the far case 
this holds iff $d(s,t,e) = d(s,x,e) + d(x,t)$.
Otherwise, that is, if no replacement path $P(s,t,e)$ uses $x$
or, equivalently, $d(s,t,e) < d(s,x,e) + d(x,t)$, the query is said to be in the \emph{far case II}.

It takes too much space to store the replacement distances for all edges in the far case,
or memorize which edge falls in which subcase.
Instead, we build two small data structures and, at query time, compute
two (potentially different) distances.
We show that always the smaller one is correct, 
which we return as the final answer.
First, for every pivot $x \in D$ and edge $e \in P(s,x)$, 
we store the replacement distance $d(s,x,e)$.
Since $|D| \le \sqrt{n}$ and $|E(P(s,x))| \le n$, we can do so in $O(n^{3/2})$ time and space.
Given a query $(t,e)$ in the far case, 
we access the storage corresponding to $D[t] = x$, retrieve $d(s,x,e)$, 
and add $d(x,t) = d(s,t)-d(s,x)$.
This gives the first candidate distance.
It may overestimate  $d(s,t,e)$, namely, if $e$ belongs to the far case II.

\subparagraph*{Far case II.}

This case is more involved than the previous.
We make extensive use of what we call break points.
Let $e_1, \dots, e_k$ be the edges of $P(s,t)[s..x]$ in the far case~II (w.r.t.\ $t$)
in increasing distance from $s$.
We then have $d(s,t,e_1) \ge \dots \ge d(s,t,e_k)$.
This is due to the fact that any replacement path $P(s,t,e_{i})$
avoids the whole subpath starting with $e_i$ and ending in $x$. 
Its length is thus at least the replacement distance for any $e_j$, $j \ge i$.
Let $u_i$ be vertex of $e_i$ that is closer to $s$.
We say $u_i$ is a \emph{break point} if $d(s,t,e_{i}) > d(s,t,e_{i+1})$.
A break point is the beginning of a segment in which the edges in the far case II 
have equal replacement distance.
We show that there are only $O(\sqrt{Mn})$ break points/replacement distances.

For the analysis, we let the edges choose a \emph{representative} replacement path.
They do so one after another in the above order.
Edge $e_i$ first checks whether one of its replacement paths has
previously been selected by an earlier edge $e_h$, $h < i$.
If so, it takes the same one; otherwise, it chooses a possible replacement path arbitrarily.
Let $\Rep$ denote the set of representatives and
let $R \in \Rep$.
We define $z_R$ to be the first vertex on the detour part of $R$.
The vertices $z_R$, $R \in \Rep$, are also important for the subquadratic algorithm in \autoref{sec:subquadratic_preprocessing}.

\begin{restatable}{lemma}{representatives}
\label{lem:structure_of_reps}
	Edges $e_i$ and $e_j$ that belong to the far case II
	choose the same representative iff $d(s,t,e_i) = d(s,t,e_j)$.
	All representatives have different lengths and $|\Rep|$ equals the number of break points.
	Let $R,R' \in \Rep$ be such that $R'$ is the next shorter representative after $R$.
	We have $d(s,z_{R}) < d(s,z_{R'})$
	and all edges represented by $R$ lie on the subpath $P(s,t)[z_R .. z_{R'}]$.
	There is exactly one break point on $P(s,t)[z_R .. z_{R'}]$,
	the one corresponding to length $w(R)$.
\end{restatable}

\begin{proof}
	Edges with different replacement distances have 
	disjoint sets of replacement paths to choose from.
	Now suppose the replacement distances  $d(s,t,e_i) = d(s,t,e_j)$ are equal.
	Without loosing generality, the edge $e_j$, $j \ge i$, is further away from $s$
	and selects its representative after $e_i$.
	The representative replacement path $R$ for edge $e_i$ also avoids $e_j$
	since it does not remerge with $P(s,t)$ prior to pivot $x$.
	As the distances $d(s,t,e_j) = d(s,t,e_i) = w(R)$ are the same,
	$R$ is in fact a replacement path for $e_j$
	and is selected again as representative.
	The assertions of the different lengths and the total number of representatives easily follow.
	
	Let $R \in \Rep$ be a representative replacement path.
	The first vertex $z_R$ of its detour part must be closer to $s$
	than all edges it represents as $R$ avoids them.
	Let $e^*$ be the edge closest to $s$ that belongs to the far case II
	and is represented by $R$.
	The break point $u^*$ starting the segment with replacement distance $w(R)$
	is thus the vertex of $e^*$ that is closer to $s$.
	
	Let now $R' \in \Rep$ be the next shorter representative after $R$.
	If we had $d(s,z_{R'}) \ge d(s,z_R)$,
	then $R'$ would be a path that avoids $e^*$ and has length $w(R') < w(R) = d(s,t,e^*)$
	strictly smaller than the replacement distance, a contradiction.
	Reusing the same arguments as before, 
	we also get that the break point corresponding to $w(R')$
	lies after $z_{R'}$
	and that the break point $u^* \in e^*$ cannot lie below $z_{R'}$ (on subpath $P(s,t)[z_{R'}..t]$).
	In summary, the subpath $P(s,t)[z_R..z_{R'}]$ contains exactly one break point, namely, $u^*\!$.
\end{proof}

It is left to prove that $|\Rep| = O(\sqrt{Mn})$.
The following lemma is the heart of our compression scheme.
It simplifies and thereby generalizes a result
by Gupta and Singh~\cite{GuptaSingh18FaultTolerantExactDistanceOracle} for unweighted undirected graphs.
The argument we use is versatile enough to not only cover integer-weighted graphs,
it extends to vertex failures as well (\autoref{lem:far_case_II_vertex_failures}).
A similar idea also allows us to design an oracle with constant query time
(\autoref{lem:constant_query_time}) and the subquadratic preprocessing algorithm
on sparse graphs (\autoref{thm:subquadratic_preprocessing}).
Unfortunately, the argument crucially depends on the graph being undirected.
New techniques are needed to compress the fault-tolerant distance information in directed graphs.

\begin{lemma}
\label{lem:far_case_II_edge_failures}
	The number of representatives for edges on $P(s,t)$ is 
	$|\Rep| \le 3 \sqrt{Mn}$.
\end{lemma}

\begin{proof}
	All representatives are of different length by \autoref{lem:structure_of_reps}.
	Also, they have length at least $d(s,t)$,
	the weight of the original $s$-$t$-path $P$.
	Hence, there are only $2\sqrt{Mn}$ many of length at most $d(s,t) + 2\sqrt{Mn}$.
	We now bound the number of \emph{long} 
	representatives, which are strictly longer than that.
	Let $R$ be a long representative.
	Its detour part is longer than $2\sqrt{Mn}$, 
	whence it must span at least $2 \sqrt{n/M}$ vertices.
	Consider the path on the first $\sqrt{n/M}$ vertices of the detour
	starting in $z_R$, we call it $\Stub_R$.	
	If $\Stub_R$ does not intersect with
 	$\Stub_{R'}$ for any other long $R' \in \Rep$, $R' \neq R$,
	there can only be $n/(\sqrt{n/M}) = \sqrt{Mn}$ stubs in total 
	and thus as many long representatives.
	
	To reach a contradiction, assume the stubs of $R$ and $R'$ intersect. 
	Let $e$ be an edge represented by $R$ and
	$y \in V(\Stub_R) \cap V(\Stub_{R'})$ a vertex on both stubs.
	W.l.o.g.\ $R'$ is strictly shorter than $R$ and thus $z_{R'}$ comes behind $z_R$ on the path $P$
	and $e$ is on $P[z_R..z_{R'}]$ (\autoref{lem:structure_of_reps}).
	Note that $w(\Stub_R), w(\Stub_{R'}) \le \sqrt{Mn}$.	
	Therefore, the path
	$P^* = P[s..z_R] \circ R[z_R ..y] 
		\circ R'[y.. z_{R'}] \circ P[z_{R'}  ..t]$
	avoids $e$ and has length
	$w(P^*) \le d(s,t) + w(R[z_R  ..y]) + w(R'[y.. z_{R'}]) 
		\le d(s,t) + 2 \sqrt{Mn} < w(R)$.
	This is a contradiction to $R$ being the representative of $e$.
\end{proof}

\noindent
Observe how the argument in the proof above depends on the fact
that we can traverse the segment $R'[z_{R'} ..y] \subseteq \Stub_{R'}$ in both directions.
When following $R'$ from $s$ to $t$, we visit $z_{R'}$ prior to $y$, 
while for $P^*$ it is the other way around.
This is not necessarily true in a directed graph. 
Indeed, one can construct examples that have a directed path on $\Omega(n)$ edges
in which each of them has its own replacement distance.

With access to the replacement distances,
all break points can be revealed by a linear scan
of the path $P(s,t)$ in time $O(n)$.
Let $z_{i_1}, \dots, z_{i_{|\Rep|}}$ be the break points
ordered by increasing distance to the source $s$
and $e_{i_1}, \dots, e_{i_{|\Rep|}}$ the corresponding edges.
For the data structure, we compute an ordered array of the original distances 
$d(s,z_{i_1}) < \dots < d(s,z_{i_{|\Rep|}})$
associated with the replacement distances $d(s,t,e_{i_j})$, taking $O(\sqrt{Mn})$ space.
Let $(e, t)$ be a query with $e = \{u,v\}$.
We compute the index $j = \argmax_{1 \le k \le |\Rep|} \nwspace \{ d(s,z_{i_k}) \le d(s,u) \}$,
with a binary search on the array in $O(\log n)$ time
and retrieve $d(s,t,e_{i_j})$ as the second candidate distance.

The edge $e$ lies on the subpath $P(s,t)[z_{i_j}..z_{i_{j+1}}]$
(respectively, on $P(s,t)[z_{i_{|\Rep|}}..x]$ if $j = |\Rep|$).
It thus has replacement distance \emph{at most} $d(s,t,e_{i_j})$.
If $e$ belongs to the far case II, the second candidate distance is exact and (strictly)
smaller than the first one $d(s,x,e)+d(x,t)$;
otherwise, the first candidate is smaller (or equal) and correct.
 
Scaling this solution to all targets $t \in V$ gives a total space requirement of
$O(M^{1/2} \nwspace n^{3/2})$.
However, the preprocessing time is $O(n^2)$, dominated by the linear scans for each target.

\subsection{Extensions}
\label{subsec:reduction_extension}

There are several possible extensions for our Single-Source DSO.
While the transfer to vertex failures comes for free,
reducing the query time to a constant, making the oracle path-reporting,
or returning the whole fault-tolerant shortest path tree incurs additional
costs of a higher space requirement or preprocessing time, respectively.
We still assume the setting
of \autoref{thm:reduction}, i.e., oracle access to the replacement distances
for failing edges/vertices.

\textbf{\textsf{Vertex failures.}}
The solutions for the near case, and far case I
hold verbatim also for vertex failures.
A vertex on the path $P(s,t)[s..x]$, except $s$ itself,
belongs to the \emph{far case~II} iff it satisfies
$d(s,t,v) < d(s,t,x) + d(x,t)$.
Let $\Rep_V$ be the sets of representatives,
now chosen by the \emph{vertices}.
The advantage of the proof of \autoref{lem:far_case_II_edge_failures}
is that it easily transfers to vertex failures.
While the stubs of the detours may no longer be unique,
they now intersect at most one other stub
and identify \emph{pairs} of representatives. 

\begin{restatable}{lemma}{vertexrepresentatives}
\label{lem:far_case_II_vertex_failures}
	The number of representatives for vertices on $P(s,t)$ is
	$|\Rep_V| \le 5  \sqrt{Mn}$.
\end{restatable}

\textbf{\textsf{Constant query time.}}
If we could query the break point of an edge in the far case II in $O(1)$ time,
our Single-Source DSO had a constant overall query time.
However, since the break points also depend on the target $t$, hard-coding them
would yield a $O(n^2)$ space solution, which is wasteful for $M = o(n)$.
Instead, we improve the analysis in \autoref{lem:far_case_II_edge_failures}.
It hardly made any use of the fact that the pivot $x$ is among the last $\sqrt{n}$ vertices
on the $s$-$t$-path in $T_s$ and considered only a single target.
We now strike a balance between selecting more pivots
and grouping targets with the same assigned pivot together.

\begin{restatable}{lemma}{constantquerytime}
\label{lem:constant_query_time}
	There is an algorithm that,
	when given oracle access to the replacement distances for failing edges (vertices),
	preprocesses in $O(n^2)$ time a Single-Source DSO for edge (vertex) failures
	taking $O(\min \{M^{1/3} \nwspace n^{5/3}, n^2 \})$ space
	and having constant query time.
\end{restatable}

\textbf{\textsf{Path-reporting oracles.}}
We can adapt our Single-Source DSOs to also report the replacement paths using the same space.
However, to do so it is not enough to have access to the replacement distances
as the paths depend on the structure of $G$.
Also, making the oracle path-reporting increases in preprocessing time,
which now also depends on $m$.

\begin{restatable}{lemma}{pathreporting}
\label{lem:path-reporting}
	With access to $G$,
	there is a path-reporting Single-Source DSO for edge (vertex) failures
	with $O(\min \{m \sqrt{Mn}, mn \} + n^2)$ preprocessing time and either
	$O(\min \{M^{1/2} \nwspace n^{3/2}, n^2 \})$ space and $\Otilde(1)$ query time per edge,
	or $O(\min \{M^{1/3} \nwspace n^{5/3}, n^2 \})$ space and $O(1)$ query~time.
\end{restatable}

\textbf{\textsf{Fault-tolerant shortest path tree oracles.}}
We are going one step further in the direction of fault-tolerant subgraphs,
see for example~\cite{Bilo18FTApproximateSPT,PaPeP16}.
We enable our oracle to report, for any failing edge or vertex,
the whole fault-tolerant single-source shortest path tree.
Compared to the path-reporting version, we make sure to return every tree edge only once.

\begin{restatable}{lemma}{faulttoleranttree}
\label{lem:fault-tolerant_tree}
	With access to $G$,
	there is a data structure with $O(\min \{m \sqrt{Mn}, mn \} + n^2)$ preprocessing time, taking
	$O(\min \{M^{1/2} \nwspace n^{3/2
	}, n^2 \})$ space 
	that, upon query $e \in E$ (respectively, $v \in V$),
	returns a shortest path tree for $G-e$ (respectively, $G-v$) rooted in $s$ in time $O(n)$.
\end{restatable}

\section{Space Lower Bound}
\label{sec:space_lower_bound}

We now present an information-theoretic lower bound showing
that the space of the Single-Source DSO
resulting from our reduction is optimal up to the word size.

\spacelowerbound*

\begin{proof}
	Let $M' = \min \{M,n\}$.
	We give an incompressibility argument
	in that we show that one can store any binary $n \times n$ matrix
	$X$ across $\sqrt{n/M'}$ Single-Source DSOs.
	Not all of them can use only $o( \sqrt{M'} \nwspace n^{3/2})$ bits of space
	as otherwise this would compress $X$ to $o(n^2)$ bits.
	We create graphs $G_1, G_2, \dots, G_{\sqrt{n/M'}}$.
	Each of them has $O(n)$ vertices and maximum edge weight $M$.
	The graph $G_k$ will be used to store the $\sqrt{M' n}$ rows of $X$
	with indices from $(k-1)\sqrt{M'n} + 1$ to $k \sqrt{M'n}$.
	
	We first describe the parts that are common to all of the $G_k$.
	Let $A = \{a_1, \dots, a_n\}$ and $B = \{b_1, \dots, b_n\}$ be two sets of $n$ vertices each,
	we connect $a_i$ and $b_j$ by an edge of weight 1 iff $X[i,j] = 1$.
	There are no other edges between $A$ and $B$.
	We also add a path $P = (v_1, \dots, v_{\sqrt{M'n}})$ all of whose edges have weight $1$.
	The vertex $s = v_{\sqrt{M'n}}$ is the source in each graph.
	Also, let $\{v_0, v_1\}$ be an edge of weight $M$,
	it serves to raise the maximum edge weight to $M$, if needed.
	Specifically in $G_k$ and for each $1 \le i \le \sqrt{M'n}$,
	we connect the vertex $v_i$ with $a_{(k-1)\sqrt{M'n} + i}$ by a path $P_{k,i}$
	of total weight $2i-1$.
	Due to the edge weights, we can make the path $P_{k,i}$ so that it
	uses at most $2i/M'$ edges and thus so many new vertices.
	In total, $G_k$ has at most 
	$2n + (\sqrt{M'n}+1) + \sum_{i=1}^{\sqrt{M'n}} \frac{2i}{M'} = O(n)$ vertices
	due to $M' \le n$.
	
	Let $e_i$ denote the edge $\{v_{i-1}, v_i\}$ on $P$.
	We claim that $X[(k{-}1)\sqrt{M'n} + i, \nwspace j] = 1$
	if and only if the replacement distance in $G_k$
	is $d_{G_k}(v_{\sqrt{M'n}}, b_j, e_i) = \sqrt{M'n} + i$.	
	We assume $k = 1$, larger $k$ follow in the same fashion.
	Observe that one has to go through a vertex in $A' = \{a_i, a_{i+1}, \dots, a_{\sqrt{M'n}}\}$ to reach $b_j$ from the source $s = v_{\sqrt{M'n}}$.
	Conversely, $A'$ is the only part of $A$ that
	is reachable from $s$ in $G_1\,{-}\,e_i$ without using any vertex of $B$.
	
	If there is no replacement path from $s$ to $b_j$ avoiding $e_i$,
	we have $d_{G_1}(s, b_j, e_i) = \infty$ and $X[i',j] = 0$
	for all $i \le i' \le \sqrt{M'n}$, as desired.
	Let thus $P(s,b_j,e_i)$ be a replacement path 
	and further $a_{i^*}$ its first vertex that is in $A$ 
	(the one closest to the source $s$).
	Therefore, $i^* \ge i$ and $P(s,b_j,e_i)$
	has the form $(v_{\sqrt{M'n}}, \dots, v_{i^*}) \circ P_{1,i^*} \circ P'$
	for some $a_{i^*}$-$b_j$-path $P'$.
	It holds that $d_{G_1}(v_{\sqrt{M'n}}, b_j, e_i)
		= (\sqrt{M'n} - i^*) + (2i^*-1) + w(P') = \sqrt{M'n} + i^* - 1 + w(P') \ge \sqrt{M'n} + i$.
	Equality holds only if $i^* = i$ and $w(P) = 1$,
	thus $a_i$ must be a neighbor of $b_j$ and $X[i,j] = 1$ follows;
	otherwise, the replacement distance is strictly larger.
\end{proof}

\section{Derandomizing Single-Source Replacement Paths Algorithms}
\label{sec:derandomize}

In this section, we derandomize the combinatorial $\Otilde(m \sqrt{n} + n^2)$ time
algorithm for SSRP of Chechik and Cohen~\cite{ChechikCohen19SSRP_SODA}
obtaining the same asymptotic running time.
In \autoref{sec:derand-grandoni}, we also derandomize the algebraic SSRP algorithm of Grandoni and Vassilevska Williams.
When combined with the reduction of \autoref{sec:reduction_algorithm_to_DS},
they give deterministic Single-Source DSOs.

Suppose the base graph $G = (V,E)$ is unweigted.
It follows from a result by 
Afek~et~al.~\cite[Theorem~1]{Afek02RestorationbyPathConcatenation_journal}
that for every target $t \in V$, edge $e \in E$, and replacement path $P(s,t,e)$ in $G - e$,
there exists a vertex $q$ on $P(s,t,e)$ 
such that both subpaths $P(s,t,e)[s .. q]$ and $P(s,t,e)[q .. t]$
are shortest paths in the original graph $G$. 
Computing the vertex $q$ directly for each pair $(t,e)$ is too expensive.
Instead, the algorithm in~\cite{ChechikCohen19SSRP_SODA} employs a random hitting set for the subpaths.
The only randomization used in~\cite{ChechikCohen19SSRP_SODA} is
to sample every vertex independently 
with probability $O((\log n)/\sqrt{n})$ to create a set $B \subseteq V$ of so-called \emph{pivots}.
The set $B$ contains $\Otilde(\sqrt{n})$ such pivots w.h.p.
The correctness of the algorithm relies on the following important property. 
With high probability,
there exists a vertex $x \in B \cup \{ s \}$ before $q$ on $P(s,t,e)$
and a vertex $y \in B \cup \{t\}$ after $q$
such that the subpath of $P(s,t,e)[x..y]$ has length only $\Otilde(\sqrt{n})$.
Here, we describe how to compute the set $B$ deterministically with the same properties.
We defer the proof of correctness of the algorithm to \autoref{appendix:derand-soda-proof}. 

We derandomize the vertex selection using an approach similar to the one of 
Alon, Chechik, and Cohen~\cite{AlonChechikCohen19CombinatorialRP}.
Given paths $D_1, \ldots, D_k$, where each contains at least $L$ vertices,
the folklore greedy algorithm constructs a hitting set of size $\Otilde(n/L)$,
by iteratively covering the maximum number of unhit paths, in $\Otilde(kL)$ time.
The challenge is to quickly compute a suitable set of paths.
We construct three systems of path $\mathcal{L}_1$, $\mathcal{L}_2$, and $\mathcal{L}_3$
to obtain $B$.

We prepare some notation.
For a rooted tree $T$, a vertex $v \in V(T)$, and an integer parameter $L \ge 0$, let $\Last_{T, L}(v)$ be the subpath containing the last $L$ edges of the path in the tree $T$ from the root to $v$,
or the whole path if it has length less than $L$.
Let $|\Last_{T_s, L}(v)|$ denote the number of edges on the path.

\begin{itemize}
	\item \textbf{\textsf{Paths $\mathcal{L}_1$ and hitting set $B_1$.}}
	Set $\mathcal{L}_1$ contains the last $\sqrt{n}/2$ edges of every path in $T_s$,
	$\mathcal{L}_1 = \{ \Last_{T_s,  \sqrt{n}/2}(v) \ | \ 
		v \in V, \ |\Last_{T_s, \sqrt{n} / 2}(v)| = \sqrt{n}/2  \}$.
	As an alternative, we can also use \autoref{lem:set_of_pivots}
	to compute in $\Otilde(n)$ time a deterministic hitting set $B_1$ for $\mathcal{L}_1$ 
	of size $2\sqrt{n}$. 	
	\item \textbf{\textsf{Paths $\mathcal{L}_2$ and hitting set $B_2$.}}
	We run a breadth-first search from every vertex $x \in B_1$ 
	to compute the shortest paths trees $T_x$ rooted in $x$,
	and define the second set to be
	$\mathcal{L}_2 = \{ \Last_{T_x, \sqrt{n} / 2}(y) \ | \ 
		x  \,{\in}\, B_1 ,y \,{\in}\, V, \ |\Last_{T_x, \sqrt{n} / 2}(y)| = \sqrt{n} / 2 \}$
	Greedy selection computes a hitting set $B_2$ for $\mathcal{L}_2$
	of size $\Otilde(\sqrt{n})$ in total time $\Otilde(n^2)$.
\end{itemize}

Before we can define $\mathcal{L}_3$, we need additional notation.
Let $e = \{u,v\}$ be an edge in $T_s$ such that $u$ is closer to $s$ than $v$
and let $T_{s,v}$ be the subtree of $T_s$ rooted in $v$.
Let further $G_e = (V_e, E_e, w_e)$ be a weighted graph
such that $V_e$ contains $s$ and the vertices $x \in V(T_{s,v})$
with $d(s,x) \le d(s,v) + 4\sqrt{n}$. 
The edges of $G_e$ that are inside of $T_{s,v}$ are the same as in $G$,
and additionally every shortest path $P(s,x)$ from $s$ to every vertex $x \in V_e$ such that $P(s,x)$ passes only through vertices outside of $V_e$ (except for its first vertex $s$ and its last vertex $x \in V_e$) is  replaced with a shortcut edge $(s,x)$ whose weight is equal to the length $d(s,x)$ of the corresponding shortest path $P(s,x)$, preserving the original paths distances (using weights). 
The SSRP algorithm in~\cite{ChechikCohen19SSRP_SODA} 
computes Dijkstra's algorithm from $s$ in each $G_e$.
We let $T_{G_e}$ denote the resulting shortest path tree.
For more details, see \autoref{subapp:constructing_G_e}.

\begin{itemize}
	\item \textbf{\textsf{Paths $\mathcal{L}_3$ and hitting set $B_3$.}}
	The third set
	$\mathcal{L}_3 := \{ \Last_{T_{G_e},  \sqrt{n} / 2 }(x) \ | \ x \in V,
		e \,{\in}\, E(T_s), \\ |\Last_{T_{G_e}, \sqrt{n} / 2 }(x)| = \sqrt{n} / 2 \}$
	contains $O(n^{3/2})$ paths as every vertex $x \in V$ belongs to at most $4\sqrt{n}$ graphs $G_e$.
	We thus get a hitting set $B_3$ of size $\Otilde(\sqrt{n})$ in time $\Otilde(n^2)$.
\end{itemize}

\noindent
The deterministic set $B = B_1 \cup B_2 \cup B_3$
can then be used as pivots in the SSRP algorithm. 

\section{Subquadratic Preprocessing on Sparse Graphs}
\label{sec:subquadratic_preprocessing}

Finally, we show how to obtain a Single-Source DSO with subquadratic preprocessing  
at least on sparse graphs. 
In order to prove \autoref{thm:subquadratic_preprocessing},
we present an algorithm running in time 
$\Otilde(M^{7/8} \nwspace m^{1/2} \nwspace n^{11/8} +\frac{M^{1/8} \  m^{3/2}}{n^{3/8}})$.
If $m = O(M^{3/4} \nwspace n^{7/4})$, then the dominating term is 
$\Otilde(M^{7/8} \nwspace m^{1/2} \nwspace n^{11/8})$.
If the graph even satisfies $m=O(n^{5/4-\varepsilon}/M^{7/4})$ for any $\varepsilon > 0$,
then the preprocessing time is $\Otilde(n^{2-\varepsilon/2})$.
We explain the main part of the randomized algorithm that allows us to design the Single-Source DSO.
The algorithm is easily adaptable to deal with vertex failures as well.
The proofs and some of the technical details 
are deferred to \autoref{app:subquadratic} due to the lack of space. 
In the following, we assume that the graph is indeed sparse, 
that is, $m=O(n^{5/4-\varepsilon}/M^{7/4})$.
%
%
The next sampling lemma is folklore, see e.g.~\cite{GrandoniVWilliamsFasterRPandDSO_journal,RodittyZwick12kSimpleShortestPaths}.

\begin{restatable}{lemma}{sampling}\label{lemma:sampling}
	Let $H$ be a graph with $n$ vertices, $c > 0$ a positive constant,
	and $L$ such that $L \ge c \ln n$.
	Define a random set $R \subseteq V$ by sampling each vertex to be in $R$ independently
	with probability $(c \ln n)/L$.
	Then, with probability at least $1-\frac{1}{n^c}$, the size of $R$ is $\Otilde(n/L)$.
	Let further ${\mathcal P}$ be a set of $\ell$ simple paths in $H$, each of which spans at least $L$ vertices.
	Then, with probability at least $1-\frac{\ell}{n^c}$, we have $V(P) \cap R \neq \emptyset$ for every $P \in {\mathcal P}$.
\end{restatable}

We employ random sampling to hit one shortest path on at least 
$L=\frac{n^{11/8}}{M^{1/8} \  m^{1/2}}$ edges for every pair of vertices.
Any vertex is included in the set $R$ of \emph{random pivots}
independently with a probability of $(3 \ln n)/L$.
We also include the source $s$ in $R$ to hit all short $s$-$t$-paths. 
By \autoref{lemma:sampling}, we have
$|R|=\Otilde(n/L)=\Otilde\big(\frac{M^{1/8} \ m^{1/2}}{n^{3/8}} \big)$ w.h.p. 
Randomization is used here since it takes too long 
to handle the $O(n^2)$ paths explicitly.

We additionally construct a set $D$ of (possibly different, regular) \emph{pivots}
that are used to classify replacement paths into near case, far case I, and far case II 
similar to \autoref{sec:reduction_algorithm_to_DS}.
The set $D$ is computed deterministically using \autoref{lem:set_of_pivots}, where we select a pivot every $\sqrt{n}$ levels. 
For a target vertex $t \neq s$, the \emph{proper pivot} of $t$
shall be that pivot $x \in D$ on the path $P(s,t)$ in $T_s$ that is closest to $t$
but satisfies $d(x,t)\geq 4ML$, or $x = s$ if there is no such pivot.
We let $D_1[t]$ denote the proper pivot of $t$
and $D_2[t] = D_1[ \nwspace D_1[t] \nwspace]$, provided that $D_1[t] \neq s$.

For every random pivot $\chi \in R$ and every edge $e$ on the path $P(s,\chi)$, we compute $d(s,\chi,e)$  in $\Otilde(m)$ time per pivot using the algorithm of Malik, Mittal, and Gupta~\cite{MaMiGu89}.
In the same time bound, we also get the vertex of $P(s,\chi)$ at which $P(s,\chi,e)$ diverges and we assume that $P(s,\chi,e)$ represents the path that diverges from $P(s,\chi)$ at a vertex that is as close as possible to $s$.\footnote{%
	The replacement path $P(s,\chi,e)$ computed in~\cite{MaMiGu89} is obtained 
	as the concatenation of a subpath $P(s,u)$ of $T_s$, an edge $\{u,v\}$ of $G-e$, 
	and a subpath $P(v,\chi)$ in $T_{\chi}$ (the shortest paths tree of $G$ rooted at $\chi$).
}
For each pivot $x \in D$ and every $e$ on $P(s,x)$,
we also compute $d(s,x,e)$. 
This takes total time 
$\Otilde(m \nwspace (|D| + |R|))
	= \Otilde\big(m n^{1/2} + \frac{M^{1/8} \ m^{3/2}}{n^{3/8}} \big)
	= \Otilde\big( \frac{m^{1/2} \nwspace n^{9/8 - \varepsilon/2}}{M^{7/8}} + \frac{M^{1/8} \ m^{3/2}}{n^{3/8}} \big)$
and allows us to answer replacement distance queries 
in $O(1)$ time if the target is in $D \cup R$.

We are left to handle non-pivot targets.
Fix a $t \in V{\setminus}(D \cup R)$ and let $x_1=D_1[t]$, and $x_2=D_2[t]$. 
We use similar cases as before.

\begin{itemize}
\item \textbf{\textsf{Near case.}} The edge $e$ is on $P(s,t)[x_2 ..t] = P(x_2,t)$.
\item \textbf{\textsf{Far case I.}} The edge $e$ is on $P(s,t)[s..x_2] = P(s,x_2)$
	and there is a replacement path $P(s,t,e)$ that passes through $x_2$.
\item \textbf{\textsf{Far case II.}} The edge $e$ is on $P(s,x_2)$ 
	and there is no replacement path $P(s,t,e)$ that passes through $x_2$.
\end{itemize}

\noindent
In the remainder, we show how to efficiently compute the replacement distances in the far case II
as previously this was the only case with quadratic run time.
The technical details of the near case are reported in \autoref{app:subquadratic}.
A shortest path tree of $G$
and the replacement distances to targets in $D$ are enough to handle the far case I,
see \autoref{sec:reduction_algorithm_to_DS}.

Since in the far case II the pivot $x_2$ lies on $P(s,t)$,
we can assume $P(s,t)$ to have length $d(s,t) \ge d(x_2,t) \ge 4ML$ and at least $4L$ edges.
In the following, we use different indexing schemes pointing
to objects and distances related to $P(s,t)$,
all of them are ordered from the source $s$ to pivot $x_2$.
First, we denote by $R_1,\dots, R_k$ the $k$ representative replacement paths 
for edges in the far case II.
We have $k \leq 3\sqrt{Mn}$ by \autoref{lem:far_case_II_edge_failures}. 
Let the \emph{distinguished} edge $e^*_\ell \in P(s,t)$ be the one
that is closest to $s$ such that $R_\ell$ represents $e^*_\ell$, i.e., $R_\ell$ is a replacement path in $G-e^*_\ell$ and we fall in far case II.
Set $d_\ell = w(R_\ell)$.
As no replacement path from $s$ to $t$ for edge $e_\ell^*$ uses vertex $x_2$,
we have $d_\ell < d(s,x_2,e_\ell^*)+d(x_2,t)$.
The distinguished edges $e^*_1, \dots, e^*_k$ are ordered by increasing distance from $s$,
this implies $d_1 > \dots > d_k$ for their replacement distances, see \autoref{sec:reduction_algorithm_to_DS}.
Furthermore, let $N$ be the number of \emph{all} edges (of the far cases I and II)
on the path $P(s,x_2) = (e_1, e_2, \dots, e_N)$,
seen in order from $s$ to $x_2$. 
This way, we identify $P(s,x_2)$ with the interval $[1,N]$,
an index $j \in [1,N]$ stands for the $j$-th edge $e_j$ on $P(s,x_2)$.
With a slight abuse of notation, we also say that $e_j \in [a,b]$ in case $j \in [a,b]$.

We employ the random pivots to efficiently compute all the $k$ pairs $(d_\ell,e^*_\ell)$ w.h.p.
The key idea is that, for each failing edge $e$ on $P(s,x_2)$,
there exists w.h.p.\ a random pivot $\chi \in R$ such that $d(\chi,t) \leq ML$
and $d(s,t,e)=d(s,\chi,e)+d(\chi,t)$ simultaneously hold.
To see this, recall that any replacement path $P(s,t,e)$ has at least $4L$ edges 
and let $y$ be the vertex such that $P(s,t,e)[y..t]$ consists of the last $L$ of them.
We claim that $P(s,t,e)[y..t]$ is in fact a shortest path in $G$.
Assume there were a shorter $y$-$t$-path,
then it must contain $e$ and have length at least $d(x_2,t) \ge 4ML$, a contradiction.
Therefore, \emph{some} shortest $y$-$t$-path in $G$ has at least $L$ edges
and is thus hit by a random pivot $\chi$ w.h.p.,
which gives the equality.
Any reference to high probability refers to this fact. 
We use it to design a recursive algorithm
that finds the pairs $(d_\ell,e^*_\ell)$ w.h.p.\ in time 
$O(|R| \nwspace M^{3/4} \nwspace n^{3/4})= \Otilde(M^{7/8} \nwspace m^{1/2} \nwspace n^{3/8})$ 
per target.

Recall that we view $P(s,x_2)$ as $[1,N]$.
When exploring a subinterval $[a,b]$, the algorithm searches for 
a pair $(d_\ell,e^*_\ell)$ with a distinguished edge $e^*_\ell \in [a,b]$.
The algorithm knows both an upper bound $\Delta_{[a,b]}$
and a lower bound $\delta_{[a,b]}$ on the admissible values for $d_\ell$.
More precisely, $\Delta_{[a,b]}+1$ corresponds w.h.p.\ to the smallest possible value $d_{\ell'}$
such that $e^*_{\ell'} \in [1,a{-}1]$ (the lower the index, the higher is $d_{\ell'}$);
similarly, $\delta_{[a,b]}-1$ is the the largest possible value $d_{\ell'}$ for $e^*_{\ell'} \in [b{+}1,N]$.
In the beginning, we set $\Delta_{[1,N]}=\infty$, $\delta_{[1,N]}=0$ and the algorithm
explores the entire interval $[1,N]$.
It terminates when there are no more unexplored subintervals. 

We now describe the search for $d_\ell$ with $e^*_\ell \in [a,b]$.
We assume $a \leq b$ and $\delta_{[a,b]} \leq \Delta_{[a,b]}$
as otherwise no such pair exists.
Set $\mu = \max_{j \in [a,b]}\{ d(s,x_2,e_j)+d(x_2,t)\}$.
The algorithm keeps searching in the interval only if $\mu > \delta_{[a,b]}$. 
Indeed, if $\mu \leq \delta_{[a,b]}$, we know for sure that such a pair does not exist
as there must be a replacement path (of type far case I) that passes through vertex $x_2$.
We first compute the largest index $j \in [a,b]$ for which $\mu =d(s,x_2,e_j)+d(x_2,t)$.
We do so by employing a range minimum query (RMQ) data structure 
to support such queries in constant time after an $O(N) = O(n)$ time preprocessing~\cite{BenderFarachColton00LCARevisited}.
Observe that the same data structure can be reused for all the target vertices $t'$ 
for which $D_2[t']=x_2$.
It is enough that it stores the values $d(s,x_2,e)$, instead of $d(s,x_2,e)+d(x_2,t)$.
The former distances are independent of the considered target and we already computed them above.
We use only $O(|D|)$ RMQ data structures, 
which we prepare in $O(n  |D|)=O(n^{3/2})$ time.

In the following, we assume $\mu > \delta_{[a,b]}$.
We select a candidate replacement path for $e_j$ by choosing the shortest one 
that runs through a random pivot in $O(|R|)$ time via brute-force search 
in the data we computed above for the targets in $R$.
Ties are broken in favor of the replacement path $P(s,\chi,e_j)$ 
that diverges from $P(s,x_2)$ at the vertex that is closest to $s$. 
Let $\delta=\min_{\chi \in R}\big\{d(s,\chi,e_j)+d(\chi,t)\big\}$ 
be the length of such a replacement path,
w.h.p.\ it is the actual replacement distance $P(s,t,e_j)$.
Let further $\chi_j$ be the minimizing random pivot,
and $z_j$ the vertex of $P(s,x_2)$ at which $P(s,\chi_j,e_j)$ diverges.
We check whether $\delta < \mu$ and $\delta \leq \Delta_{[a,b]}$ holds.
If either of the two conditions is violated,
then there is no need to keep searching in the interval $[a,j]$, 
as shown in the next lemma.
In this case, the algorithm makes a recursive call on the \emph{lower} interval $[j{+}1,b]$
(the one with smaller replacement distances)
by setting $\Delta_{[j+1,b]}=\Delta_{[a,b]}$ and $\delta_{[j+1,b]}=\delta_{[a,b]}$. 
We say that the search was \emph{unsuccessful}.

\begin{restatable}{lemma}{unsuccessful}
\label{lemma:unsuccessful}
	If $\delta \geq \mu$ or $\delta > \Delta_{[a,b]}$, then, w.h.p.\
	we have $e^*_\ell \not \in [a,j]$ for all $\ell \in [k]$.
\end{restatable}

Suppose the search is \emph{successful}, 
that is, $\delta < \mu$ and $\delta \leq \Delta_{[a,b]}$.
We then use binary search techniques\footnote{%
	Let interval $[a',b'] \subseteq [a,j]$ lie entirely below $z_j$.
	We divide it into subintervals $[a',j']$ and $[j'{+}1,b']$ of roughly equal sizes 
	and check whether the maximum value returned by the RMQ data structure on query $[a',j']$
	is still larger than $\delta$. 
	If so, we recurse on the interval $[a',j']$; otherwise, on $[j'{+}1,b']$.
}
to compute in $O(\log n)$ time 
the smallest index $i \in [a,j]$ for which the edge $e_i$ lies on the subpath $P(s,x_2)[z_j..x_2]$
and $\delta < d(s,x_2,e_i)+d(x_2,t)$ holds.
The case $i=j$ is possible.
The condition on $e_i$
is such that $P(s,\chi_j,e_j)$ also avoids $e_i$,
which implies $d(s,t,e_i) \le \delta < d(s,x_2,e_i)+d(x_2,t)$.
The edge $e_i$ must belong to the far case II w.r.t.\ target $t$.
We show that in fact $(\delta,e_i)$ is w.h.p.\ the sought pair 
with $e^*_\ell \in [a,j]$ and minimum $d_\ell$.  

\begin{restatable}{lemma}{successful}\label{lemma:successful}
	Let $\ell \in [k]$ be maximal such that $e^*_\ell \in [a,j]$. 
	Then, w.h.p.\ $\delta=d_\ell$ and $e_i=e^*_\ell$.
\end{restatable}

The algorithm outputs $(\delta, e_i)$
and recurses on the lower interval $[j{+}1,b]$ with new bounds
$\Delta_{[j+1,b]}=\delta-1$ and $\delta_{[j+1,b]}=\delta_{[a,b]}$,
as well as on the \emph{upper} interval $[a,i{-}1]$,
with $\Delta_{[a,i-1]} =\Delta_{[a,b]}$ and $\delta_{[a,i-1]} =\delta+1$. 
This is justified since the edges in $[i,j]$ that belong to the far case II
are w.h.p.\ precisely the ones represented by the path
$R_{\ell}$ of length $d_\ell = \delta$.

The time needed for one target $t$
is proportional (up to a log-factor)
to the number of random pivots and the overall number of searches. 
There are $k=O(\sqrt{Mn})$ successful searches by \autoref{lem:far_case_II_edge_failures}.
The following lemma bounds the number of unsuccessful searches.

\begin{restatable}{lemma}{numberunsuccessfulsearches}
\label{lemma:number_unsuccesful_searches}
	The number of unsuccessful searches for a single target vertex is $O(M^{3/4} \nwspace n^{3/4})$.
\end{restatable}

\noindent
The algorithm computes w.h.p.\ all pairs for one target vertex in time
$\Otilde(|R| \nwspace M^{3/4} \nwspace n^{3/4}) = \Otilde(M^{7/8} \nwspace m^{1/2} \nwspace n^{3/8})$,
scaling this to all targets gives $\Otilde(M^{7/8} \nwspace m^{1/2} \nwspace n^{11/8})$.

\bibliographystyle{plainurl} 
\bibliography{SSDSO_bib}

\newpage
\appendix

\section{Omitted Proofs of \autoref{sec:reduction_algorithm_to_DS}}
\label{app:omitted_reduction_proofs}

\setofpivots*

\noindent
We prove a more general version, which we will reuse later.
In it, we require at least one pivot among the last $L \le n$ vertices
and get $|D| \le n/L$ for the set,
but the preprocessing time and the size of the data structure stays the same.
The above formulation of \autoref{lem:set_of_pivots} then follows by setting $L = \sqrt{n}$.
  
\begin{proof}
	In this proof, we use $|P| = |V(P)|$ to denote the number of vertices of some path $P$.
	We define the set $D$ iteratively.
	Each iteration starts with a subtree $T$ of $T_s$ containing the source vertex $s$.
	Initially, we have $T = T_s$.
	If all paths in $T$ starting in $s$ have at most $L$ vertices,
	include $s$ into the set and terminate.
	That is, we set $D[t] = s$ for all $t \in V(T)$.
	Otherwise, let $v$ be a leaf in $T$ whose shortest path $P(s,v)$
	has the maximum number of vertices.
	Note that $|P(s,v)| > L$ holds.
	Let $x$ be the vertex on $P(s,v)$ for which $|P(s,v)[x..v] \nwspace| = L$, 
	and let $T_x$ the subtree of $T$ rooted in $x$.
	We include $x$ in the set by defining $D[t] = x$ for all vertices in $T_x$,
	remove $T_x$ from $T$, and continue with the next round.
	
	At least $L$ vertices get assigned their pivots in $D$ in every iteration,
	thus there are only $n/L$ iterations in total and as many vertices in $D$.
	Updating all pointers of the data structure takes time in $O(n)$.
	Finding the leaves $v$ can be done using a priority (max-)queue with the number of vertices 
	$|P(s,v)|$ as key.
	Since every vertex is touched exactly once, returned as the maximum or removed from the queue,
	this takes $\Otilde(n)$ total time.
\end{proof}

\subparagraph*{Vertex failures.}

For completeness,
we repeat the selection process of the vertex representatives in $\Rep_V$ here.
Let $t$ be the target and $x$ its assigned pivot and
let $v_1, v_2, \dots, v_k$ be those vertices on $P(s,t)[s..x]$ with $d(s,t,v_i) < d(s,x,v_i) + d(x,t)$.
They are ordered such that, for all $i < j$, we have $d(s,v_i) < d(s,v_j)$.
As in the case of edge failures, this implies $d(s,t,v_i) \ge d(s,t,v_j)$.
Each $v_j$ chooses one of its replacement paths as representative
and, if available, always prefers one that has already chosen by some earlier vertex $v_i$, $i < j$.
Similar to \autoref{lem:structure_of_reps}, there is one representative for each replacement distance $d(s,t,v)$ in the far case II
and all vertices represented by some $R \in \Rep_V$ lie on the subpath $P(s,t)[z_R  .. z_{R'}]$, where $R' \in \Rep$ is the next-shorter representative 
(the vertices occur on $P(s,t)[z_R  ..x]$ if $R$ is the shortest representative).
Observe that $R$ cannot represent the vertex $z_R$, but $z_{R'}$ is possible.

\vertexrepresentatives*

\begin{proof}
	Let the notation be as in the proof of \autoref{lem:far_case_II_edge_failures}.
	In particular, let $P = P(s,t)$, $R$ be a vertex representative with $w(R) > d(s,t) + 2 \sqrt{Mn}$,
	and $\Stub_R$ be the path on the first $\sqrt{n/M}$ vertices of the detour of $R$.
	Let $R' \in \Rep_V$ be the next-smaller representative. 
	The only instance in which the argument of \autoref{lem:far_case_II_edge_failures}
	does not extend to vertex failures
	is if $R'$ also has length larger than $d(s,t) + 2 \sqrt{Mn}$, 
	$\Stub_R$ and $\Stub_{R'}$ intersect (say, in~$y$),
	and the starting vertex $z_{R'}$ is the \emph{only one} that is represented by $R$.
	Then, the path $P[s..z_R] \circ R[z_R ..y] \circ R'[y.. z_{R'}] \circ P[z_{R'} ..t]$
	may be short, but it does not avoid any vertex represented by $R$
	and therefore does not imply a contradiction.

	We fix not only this special case but give a more general argument
	that still allows us to bound the number of vertex representatives at $O(\sqrt{Mn})$.
	For this proof, we redefine a \emph{long} representative to have length larger than
	$d(s,t) + 3\sqrt{Mn}$.
	Note that the number of short representatives is at most $3 \sqrt{Mn}$.
	We claim that the stub of a long representative does not intersect those of \emph{two} others.
	Assume otherwise and let $R, R', R''$ be three long representatives with $w(R) > w(R') > w(R'')$.
	This means, we have $d(s,z_R) < d(s,z_{R'}) <  d(s,z_{R''})$.
	We first concentrate on the case where $R'$ has the intersecting stub.
	Let $y \in V(\Stub_R) \cap V(\Stub_{R'})$
	and $y'' \in V(\Stub_{R'}) \cap V(\Stub_{R''})$.
	Then, the concatenation
	\begin{equation*}
		P^* = P[s..z_R] \circ R[z_R ..y] \circ R'[y.. y''] \circ R''[y''.. z_{R''}] \circ P[z_{R''} ..t]
	\end{equation*}
	is a path of length $w(P^*) \le d(s,t) + 3\sqrt{Mn} < w(R)$ that avoids all vertices 
	that occur strictly between $z_R$ and $z_{R''}$ on $P$,
	including all vertices represented by $R$, a contradiction.
	If instead $R$ or $R''$ has the intersecting stub, 
	there exists some $v \in V(\Stub_R) \cap V(\Stub_{R''})$
	and $P[s..z_R] \circ R[z_R ..v] \circ R''[v.. z_{R''}] \circ P[z_{R''} ..t]$
	is an even shorter path avoiding all those vertices.
	
	In summary, each stub on $\sqrt{n/M}$ vertices
	is shared by at most $2$ long representatives and thus there are at most $2\sqrt{Mn}$ of them,
	implying $|\Rep_V| \le 5 \sqrt{Mn}$.
\end{proof}

\subparagraph*{Constant query time.}

\constantquerytime*

We already gave an $O(n^2)$ space solution with constant query time
in case the maximum edge weight $M$ is larger than $n$.
Therefore, we assume $M \le n$ in the following.
Only processing the far case II incurred a super-constant query time.
In order to prove \autoref{lem:constant_query_time}, we describe a more general way
to trade the number of pivots for the number of break points.
Suppose we require at least one pivot among the last $L \le n$ vertices of each path in the shortest path tree $T_s$ that starts in the source $s$.
By (the more general version of) \autoref{lem:set_of_pivots}, we get $|D| \le n/L$ pivots.
For each $x \in D$, let $V_x$
be those targets $t$ with assigned pivot $D[t] = x$.
The $V_x$ partition $V$ and the pairwise distances 
of vertices within these sets are at most $2ML$.
If some $V_x$ has more than $2L$ elements, we split it into \emph{groups} 
of size between $L$ and $2L$.
To ease notation, we also use $V_x$ to denote the group.
In total, there are $|\mathcal{G}| = O(n/L)$  many groups.

The key observation is that the analysis in \autoref{lem:far_case_II_edge_failures}
can be improved by considering a whole group of targets simultaneously.
Fix some pivot $x \in D$.
All edges belonging to the far case II with respect to some $t \in V_x$,
meaning $d(s,t,e) < d(s,x,e) + d(x,t)$, occur on the path $P(s,x)$ in $T_s$.
We let them choose a set $\Rep_x$ of representatives again.
Each edge selects one for all its possible targets in $V_x$ before it is the next edge's turn.
While the edges are ordered by increasing distance from $s$, 
the targets are ordered by decreasing distance, ties are broken arbitrarily.
In effect, each edge selects its longest representative first.
Choices always prefer available replacement paths that have been chosen before.

\begin{lemma}
\label{lem:appendix_group_of_targets}
	The number of representatives is $|\Rep_x| = O(\frac{n}{L} + ML^2)$.
	The same holds if the representatives are instead chosen by the vertices on $P(s,x)$
	belonging to the far case II.
\end{lemma}

\begin{proof}
	A representative replacement path in $\Rep_x$ that ends in target $t \in V_x$
	is said to be \emph{short} if it has length at most $d(s,t) + 5ML$;
	otherwise, it is \emph{long}.
	Since $|V_x| \le 2L$, there are only $O(ML^2)$ short representatives in all of $\Rep_x$.
	The rest of the proof is dedicated to bound the number of long representatives.
	
	Recall that the pairwise distances of targets $t,t' \in V_x$ is at most $2ML$
	as witnessed by paths $P(t,x) \circ P(x,t')$,
	which are independent of the failing edge $e \in P(s,x)$.
	Now consider a representative $R = P(s,t,e)$ for some $t \in V_x$
	that has a detour part with length at most $2ML$,
	and thus is of total length $w(R) \le d(s,t) + 2ML$.
	It can be elongated, for \emph{every} target 
	$t' \in V_x$, to a (not necessarily shortest) $s$-$t'$-path 
	that avoids $e$ and has length at most $d(s,t) + 4ML \le d(s,t') + 5ML$.
	The inequality holds due to $x$ occurring among the last $L$ 
	vertices of both paths $P(s,t)$ and $P(s,t')$, implying $d(s,t) \le d(s,t') + ML$.
	In other words, if long representatives for targets in $V_x$ exist,
	they have a detour longer than $2ML$, these detours span at least $2L$ vertices.
	Consider the stubs consisting of the paths on the first $L$ 
	vertices of the detour of each long representative.
	The stubs cannot intersect as otherwise 
	this would form a path whose detour has length at most $2ML$ 
	and avoids all edges of at least one of the participating representatives.
	Therefore, there are only $n/L$ long representatives.
	
	The analysis can be extended to vertex failures by considering pairs of long replacement paths
	as described in the proof of \autoref{lem:far_case_II_vertex_failures}.
\end{proof}

The number of representatives for a group is larger than when considering only a single target, see~\Cref{lem:far_case_II_edge_failures,lem:far_case_II_vertex_failures}.
However, it will be advantageous to have a sublinear number of groups.
The expression in \autoref{lem:appendix_group_of_targets} is minimized by $L = (n/M)^{1/3}$.
This balances the number $|\Rep_x|$ of representatives per group,
the number $|D|$ of pivots, and the number $|\mathcal{G}|$ of groups
all at $O(M^{1/3} \nwspace n^{2/3})$.
The solutions for the the near case, and the far case I
can be implemented as described in \autoref{sec:reduction_algorithm_to_DS}.
They can be computed in time/space $O(n)$, $O(nL)$, and $O(n |D|)$,  respectively,
which is $O(M^{1/3} \nwspace n^{5/3})$.

To reach constant query time also in the far case II, we proceed as follows.
Let $t$ be a target and $D[t] = x$ the closest pivot above it in $T_s$.
A linear scan from $t$ along the path $P(s,x)$ to $s$ 
reveals all break points with respect to $t$.
The edges (vertices) that belong to the far case II and lie between consecutive break points
form a segment of equal replacement distance.
For all targets in $V_x$, there are only $|\Rep_x|$ break points,
and each one of them lies in $|V_x| \le 2L$ segments (and is the beginning of one of them).
We store in each break point the $2L$ distances together with their respective targets.
Furthermore, for each edge (vertex) on $P(s,x)$ in the far case II,
we store a pointer to the nearest break point that is closer to the source $s$.
Observe that following the pointer never crosses segment borders.
For one group, the space requirement is $O(L \cdot |\Rep_x| + n) = O(n)$;
scaling this up to all groups takes
$O(n |\mathcal{G}|) = O(M^{1/3} \nwspace n^{5/3})$ space, as desired.
To answer a query $(e,t)$ (respectively, $(v,t)$) in the far case II,
we follow the pointer of $e$ (of $v$) to the nearest break point and return
the replacement distance corresponding to $t$.
The lookups can be performed in constant time.

\subparagraph*{Path-reporting oracles.}

\pathreporting*

\begin{proof}
We explain how to modify the two Single-Source DSOs so as they can also report replacement paths.
The $O(M^{1/2} \nwspace n^{3/2})$ space oracle will then
be able to return the path in $\Otilde(1)$ time per edge,
while it is constant time per edge for the $O(M^{1/3} \nwspace n^{5/3})$ space one. However, the preprocessing time for both oracles is $O(m\sqrt{Mn}+n^2)$.
We concentrate on edge failures, the solution for vertex failures is the same.

We denote by $T_{s,e}$, with $e \in E(T_s)$, a shortest paths tree of $G-e$ rooted in $s$. For each vertex $t$ and for each edge $e$ of $P(s,t)$, we denote by $\last(t,e)$ the predecessor of $t$ in $T_{s,e}$. Finally we denote by $\last(t)$ the predecessor of $t$ in $T_s$.
%
For the sake of simplifying the exposition, we assume that the pivot of every target vertex $t \neq s$ is not $t$ itself. This can be guaranteed for any pivot $x\neq s$
if we redefine $D[x]$ as the closest proper ancestor of $x$ in $T_s$ that is also a pivot.

Let $t$ be a target vertex and let $x=D[t]$. We let $E_t$ be the set of edges of $G$ that are incident to $t$ and $m_t=|E_t|$ their number.
Observe that any $\last(t,e)$ can be computed in $O(m_t)$ time by selecting the vertex $y$ of an edge $\{y,t\} \in E_t$ such that $d(s,t,e)=d(s,y,e)+w(\{y,t\})$. 
Employing this simple selection method for every $\last(t,e)$ is  requires up to  
$\sum_{e \in E(T_s)} \sum_{t \in V}m_t = O(nm)$ time.
However, for $M \le n$, we can do better.
We use the simple selection to compute $\last(t,e)$ for every edge $e \in P(x,t)$,
the edges in the near case.
This takes $O(m_t\sqrt{n})$ for each target $t$ and, thus $O(m\sqrt{n})$ time in total.
Observe that all these values can be retrieved in constant time if we store them in a $n \times 2\sqrt{n}$ table with one row for each target $t$ and one column for each of the $2\sqrt{n}$ distances from $e$ and $t$.

We use the same method also for computing $\last(t,e)$ for the $k=O(\sqrt{Mn})$ distinct replacement paths $R_1,\dots,R_k$ for $t$ that fall in the far case II. More precisely, as already discussed while presenting the Single-Source DSOs, for each $R_i$ we know an edge $e_i^*$ such that $R_i$ is a replacement path of $G-e_i^*$. We compute each $\last(t,e_i^*)$ in $O(m_t)$ time and store this information in the same data structure we used to store the value $d(s,t,e_i^*)$. In this way, once we retrieve $d(s,t,e_i^*)$, we also know $\last(t,e_i^*)$.

Now that we have this information, the path-reporting oracle that reconstructs some replacement path $P(s,t,e)$ backwards from $t$ to $s$ is straightforward. In fact, upon query $(t,e)$, we report the edge $\{t',t\}$ with $t'=\last(t,e)$, and we recurse on $(t',e)$.
Note first that if $e \notin P(s,t)$, then $t'=\last(t)$ can be retrieved in constant time.
Otherwise, the predecessor $t'$ depends on the case in which we are.

\begin{itemize}
	\item \textbf{\textsf{Near case.}} If $e \in P(x,t)$, 
		then $\last(t,e)$ can be retrieved in $O(1)$ time.
	\item \textbf{\textsf{Far case I.}} If $e \in P(s,x)$ and $d(s,t,e)=d(s,x,e)+d(x,t)$,
		we have $t'=\last(t)$ again.
	\item \textbf{\textsf{Far case II.}} If $e \in P(s,x)$ and $d(s,t,e) < d(s,x,e)+d(x,t)$, 
		then $\last(t,e)$ can be reported in the same time as $d(s,t,e)$,
		that is, $\Otilde(1)$ for the $O(M^{1/2} \nwspace n^{3/2})$ space oracle.
\end{itemize}

The extension for the $O(M^{1/3} \nwspace n^{5/3})$ space Single-Source DSO is as follows.
We compute and store the same additional information as above in $O(m\sqrt{Mn})$ time and then
explicitly compute all shortest path trees $T_{s,e}$ with an $O(n^2)$ time algorithm that visits the target vertices in any order. For each target $t$, the algorithm scans all the edges of $P(s,t)$ in order from $s$ to $t$ and, for each failing edge $e \in P(s,t)$, it computes $\last(t,e)$ in constant time. The time needed to compute each predecessor is constant because we no longer need to perform a binary search for the values $d(s,t,e)$ that fall in the far case II. Finally, for every failing edge $e$ of $T_s$ that is not in $P(s,t)$, it is enough to set $\last(t,e)=\last(t)$.

Once all the trees $T_{s,e}$ are known, we use them to compute,
for each group of size $O((n/M)^{1/3})$, the $O(M^{1/3} \nwspace n^{2/3})$ vertices $\{z_R\}_{R \in \Rep_x}$, on the path $P(s,x)$
from which the replacement paths of the far case II with respect to that group diverge
(those are different from the break points).
This can be done by simply computing, for each target vertex $t$, the $O(\sqrt{Mn})$ vertices of $P(s,x)$ from which the replacement paths that fall in the far case II diverge. By visiting the edges $e$ of $P(s,x)$ in order from $s$ to $t$, we can find the required vertices in $O(n)$ time per target.
Indeed, using LCA data structures,
we check if $x$ is an ancestor of $t$ in $T_{s,e}$, and if so, we save the edge $e$.
We then explore all saved edges in the order in which they appear on $P(s,x)$ from $s$ to $x$ and we keep a pointer to the last divergence vertex found. As we know that the remaining divergence vertices can only be in the subpath of $T_s$ from the pointer to $x$, the pointer can only advance along the path.
\end{proof}

\pagebreak

\subparagraph*{Fault-tolerant shortest paths tree oracle.}

\faulttoleranttree*

\begin{proof}
We explain how to modify the $O(M^{1/2} \nwspace n^{3/2})$ space oracle so as, given an edge $e$ in the tree $T_s$ as a query,
it reports the shortest paths tree $T_{s,e}$ of $G-e$ rooted at $s$ in $O(n)$ time.
We will not incur the $\Otilde(1)$ penalty for accessing the distances of replacement paths
in the far case II.
However, we still need to compute their starting points 
accounting for the $O(n^2)$ term in preprocessing time.
The extension to vertex failures and $M > n$ is immediate.

We use the notation of \autoref{lem:path-reporting}.
Let $x \in D$ be a pivot and $h$ be the number of target vertices that have $x$ as a pivot.
We only have to describe the storage of the $\ell = O(h\sqrt{Mn})$ values $\last(t,e)$ of the data structure we use to retrieve the predecessor of a target vertex $t$ for which $D[t]=x$ 
when the replacement path falls in the far case II. The other values $\last(t,e)$ can always be retrieved in constant time, as discussed above.

We partition these values into $O(\sqrt{Mn})$ groups, each of size $\Theta(h)$.
More precisely, let $z_1,\dots,z_\ell$ be the vertices of $P(s,x)$ at which replacement paths of far case II diverge, in order from $s$ to $x$. For each $z_i$, let $\mu_i$ be the overall number of replacement paths of far case II type for target vertices $t$ such that $D[t]=x$ and $P(s,t,e)$ diverges from $P(s,x)$ exactly at vertex $z_i$. We partition $P(s,x)$ into $p=O(\sqrt{Mn}))$ {\em segments} $S_1,\dots,S_p$ such that $h \leq \sum_{z_i \in S_j}\mu_i \leq 2h$. For each segment $S_i$ we store all the values $\last(t,e)$ of replacement paths of far case II type that diverge from $P(s,x)$ at a vertex of $S_i$. If $i > 1$ and, for some vertex $t$ with $D[t]=x$, there is no element $\last(t,e)$ associated with segment $S_i$, we store the value $\last(t,e)$ corresponding to the minimum value $d(s,t,e)$ chosen from the previous segment $S_{i-1}$.

Upon query $e$, we scan all the pivots and, for each pivot $x$, we take the group of predecessors associated with the segment, depending on $x$, that contains $e$. For each $t$ we compute the value $\last(t,e')$ where $e'$ is the edge closest to $e$ that is above $e$ (possibly $e'=e$), if it exists. The predecessor of $t$ in $T_{s,e}$ when the replacement path falls in the far  case (either I or II) is either $\last(t)$ or $\last(t,e')$, where $\last(t)$ is preferred over $\last(t,e')$ when  $\last(t,e')$ does not exist or $d(s,x,e)+d(x,t) \leq d(s,t,e')$.
\end{proof}

\section{Derandomizing the SSRP Algorithm of Chechik and Cohen}
\label{appendix:derand-soda-proof}

We show here how to construct the graphs $G_e$ efficiently 
and prove the correctness of the derandomization procedure described in \autoref{sec:derandomize}.

\subsection{Constructing the Graphs $G_e$}
\label{subapp:constructing_G_e}

Let $T_{s,v}$ be the subtree of $T_s$ rooted in $v$, and let
$\bar{T}_{s,v}$ be the tree $T_{s,v}$ truncated at depth $4
\sqrt{n}$ (i.e., we trim the subtree of $v$ at depth
$4\sqrt{n}$, and remove from it all the vertices whose distance from
$v$ is more than $4\sqrt{n}$). 
Let $E^1_e = \{ \{x,y\} \in E \mid x \in V(\bar{T}_{s,v}) \vee y \in
V(\bar{T}_{s,v}) \} \setminus \{e\}$ be the set of all edges incident to vertices in $\bar{T}_{s,v}$ (except the edge $e$ itself), we set $w_e(\{x,y\}) = 1$ for every $\{x,y\} \in E^1_e$.
Let $E^2_e = \{ \{s,y\} \mid  y \in V(\bar{T}_{s,v}) \}$ be
additional edges from $s$ to every vertex $y$ in the subtree
$\bar{T}_{s,v}$.
Let $N(v) = \{ u \in V  \mid  \{u,v\} \in E \}$ denote the open neighborhood of $v$ in $G$.
We set $w_e(\{s,y\}) = \min \{ \infty, \min_{x \in N(y) \setminus V(T_{s,v})} \{ d(s,x) + 1 \} \}$ for every $\{s,y \} \in E^2_e$.
Let $E_e = E^1_e \cup E^2_e$.

One can run Dijkstra's algorithm from $s$ in every graph $G_e$ for every edge $e \in E(T_s)$ in $\Otilde(m\sqrt{n})$ time.
It is not difficult to observe that running all these Dijkstra computations in all the graph $G_e$ takes $\Otilde(m\sqrt{n})$ time, as every vertex $z$ belongs to at most $4\sqrt{n}$ trimmed trees $\bar{T}_{s,v}$ for some vertex $v \in V$ and thus every vertex contributes its degree to at most $4\sqrt{n}$ graphs $G_e$ and thus to at most $4\sqrt{n}$ Dijkstra's computations.
Let $T_{G_e}$ be the shortest paths tree computed in the graph $G_e$.

\subsection{Proof of Correctness}
\label{subapp:proof_of_correctness}

We prove that the SSRP algorithm is correct when using the deterministically chosen set of vertices $B$ as descrbied in \autoref{sec:derandomize}. 
We need the following definition.

\begin{definition} [Replaceability of an edge]
Given two vertices $s,t \in V$, we say that an edge $e$ is \emph{$(s,t)$-replaceable} if $d(s,t,e) = d(s,t)$. In other words, $e$ is $(s,t)$-replaceable iff there exists a shortest $s$-to-$t$ path that does not pass through $e$. 
\end{definition}

Although $G$ is undirected, we use $e = (u,v)$ to indicate 
that edge $e = \{u,v\}$ is such that vertex $u$ is closer to $s$ than $v$.
Note that, if both endpoints have the same distance from $s$,
then $e$ is not contained in the tree $T_s$.
In particular, $e$ then is $(s,t)$-replaceable for every $t$.

For the randomized pivot selection,
Chechik and Cohen~\cite{ChechikCohen19SSRP_SODA} showed that each query $(s,t,e)$ 
w.h.p.\ belongs to at least one of the following cases.

\begin{enumerate}
	\item \textbf{\textsf{Replaceable edge case.}} It holds that $d(s,t,e) = d(s,t)$.
	\item \textbf{\textsf{Small fall case.}} We have $e = (u,v)$
		 and $d(s,t) < d(s,t,e) \le d(s,u) + 4\sqrt{n}$.
	\item \textbf{\textsf{Single pivot case.}} There exists a pivot $x \in B$
		 such that $(s,x,e)$ belongs to Case~1 or 2, and $e$ is $(x,t)$-replaceable.
	\item \textbf{\textsf{Double pivot case.}} There exists two pivots $x,y \in B$
		such that $e$ is $(s,x)$-replaceable, $(x,y)$-replaceable, $(t,y)$-replaceable
		and $d(s,t,e) = d(s,x) + d(x,y) + d(y,t)$.  
\end{enumerate}

\noindent
Note that in the definition of Case~4 $x = y$ is possible.
We prove the correctness of the algorithm described in \autoref{sec:derandomize} by proving that every query $(s,t,e)$ belongs to Case~1, 2, 3, or 4 with the deterministic pivot selection as well.
This is indeed sufficient to derandomize the algorithm
as the handling of the four cases in~\cite{ChechikCohen19SSRP_SODA} is deterministic.

\begin{lemma} \label{lemma:ssrp-soda-derand}
Let $(s,t,e)$ be a query, let $B$ be the deterministic set of pivots obtained as in \autoref{sec:derandomize}. Then it holds that $(s,t,e)$ belongs to Case 1, 2, 3, or 4.
\end{lemma}

To prove \autoref{lemma:ssrp-soda-derand} we need the following lemmas.

\begin{lemma}\label{lemma:edge-on-one-side-only}
Let $y \in P(s,t,e)$, if $e \in P_{T_s}(s, y)$
then $e$ is $(y,t)$-replaceable. 
\end{lemma}

\begin{proof}
Let $e \in P_{T_s}(s, y)$ and assume by contradiction that $e$ is not $(y,t)$-replaceable. Then every shortest path $P(y,t)$ contain $e$. 

Denote by $e = (u,v)$ such that $u$ is closer to $s$ than $v$. As $e \in P_{T_s}(s,y)$ it follows that $v$ is closer to $y$ than $u$. Assume by contradiction that $e \in P(y,t)$.
If $u$ appears before $v$ along $P(y,t)$ 
then $u$ is closer to $y$ than $v$, which is a contradiction. 
If $v$ appears before $u$ along $P(y,t)$ then the path  $P_{T_s}(s,y)[s,u] \circ P(y,t)[u,t]$ is a replacement path for $(s,t,e)$ that is shorter than $d(s,y) + d(y,t)$ (and $d(s,y) + d(y,t) = d(s,t,e)$ as $t \in P(s,t,e)$) which is a contradiction.
\end{proof}

\begin{lemma}\label{lemma:cycle-lemma}
Let $x \in B \cap P(s,t,e)$ such that $e = (u,v) \in P_{T_s}(s,x)$, if there exists a simple cycle $C$ such that $e \in C$ and $d(u,x) + |C| \le 4\sqrt{n}$ then $(s,x,e)$ belongs to either Case 1 or Case 2 and $(s,t,e)$ belongs to Case 3.
\end{lemma}

\begin{proof}
The path obtained from $P_{T_s}(s,x)$ by replacing the edge $e$ with $C - \{e\}$ ({\sl i.e.}, the path $P(s,u) \circ P_{C - \{e\}}(u,v) \circ P(v,x)$) is an $s$-to-$x$ path that avoids $e$ and its length is less than $d(s,x) + |C|$.
Thus, $d(s,x,e) < d(s,x) + |C| = d(s,u) + d(u,x) + |C| \le d(s,u) + 4\sqrt{n}$ (where the last inequality holds as we assume in the lemma that $d(u,x) + |C| \le 4\sqrt{n}$). 
Then either $d(s,x,e) = d(s,x)$ and then $(s,x,e)$ belongs to Case 1, or $d(s,x) < d(s,x,e) < d(s,u) + 4\sqrt{n}$ and then $(s,x,e)$ belongs to Case 2.

Since $e \in P_{T_s}(s,x)$ then according to \autoref{lemma:edge-on-one-side-only} it holds that $e$ is $(x,t)$-replaceable.
As $x \in B$ is a pivot such that $(s,x,e)$ belongs to either Case 1 or Case 2 and $e$ is $(x,t)$-replaceable, and thus $(s,t,e)$ belongs to Case 3. 
\end{proof}

\begin{proof}[Proof of \autoref{lemma:ssrp-soda-derand}]
Assume that $(s,t,e)$ does neither belong to Cases 1,2, nor 3. We prove that it then must belong to Case~4.
Let $P(s,t,e) = ( s = v_1, v_2, \ldots, v_k = t )$ be a replacement path for $(s,t,e)$ that contains a maximum common prefix with $T_s$. More precisely, let $q$ be the last vertex along $P(s,t,e)$ such that the path from $s$ to $q$ in $T_s$ does not contain $e$. When we say that $P(s,t,e)$ 
has a maximum common prefix with $T_s$ we mean that $P(s,t,e)[s,q]$ has maximum length among all the replacement paths for $(s,t,e)$. 

By definition of $B_1$ it holds that there exists a vertex $x \in B_1 \cup \{s\}$ that hits the path $\Last_{T_s,  \sqrt{n}/2 }(q)$. 
Let $y'$ be the $\sqrt{n}$-th vertex along $P(s,t,e)[x,t]$, or if $|P(s,t,e)[x,t]| < \sqrt{n}$ then let $y' = t$. 
We consider two alternatives, either $e \in P_{T_x}(x,y')$ or $e \notin P_{T_x}(x,y')$.

We prove that it cannot be that  $e \in P_{T_x}(x,y')$.
Assume by contradiction that $e \in P_{T_x}(x,y')$, then  $P_{T_x}(x,y') \cup P(s,t,e)[x,y']$ contains a simple cycle $C$ such that $e \in C$ and $|C| \le d(x,y') + d(x,y',e) \le 2 d(x,y',e) \le 2 \sqrt{n}$ (where the last inequality holds as $y'$ is the $\sqrt{n}^{\text{th}}$ vertex along $P(s,t,e)[x,t]$ or if $|P(s,t,e)[x,t]| < \sqrt{n}$ then $y' = t$).
If $y' = t$ then according to \autoref{lemma:cycle-lemma} it holds that $(s,t,e)$ belongs to either Case 1 or Case 2. 
Assume $y' \ne t$.
By definition of $B_3$ it holds that there exists a vertex $y \in B_3$ that hits the path $\Last_{T_{G_e},  \sqrt{n}/2 }(y')$. 
Since $P(s,t,e)$ is a replacement path whose common prefix with $T_s$ is maximal, it must hold that $e \in  P_{T_s}(s,y)$ (otherwise there is a replacement path for $(s,t,e)$ whose prefix is $P_{T_s}(s,y)$ which is longer than the prefix $P_{T_s}(s,q)$ of $P(s,t,e)$). 
We obtain that $y \in B \cap P(s,t,e)$ such that $e = (u,v) \in P_{T_s}(s,y)$ and there exists a simple cycle $C$ such that $e \in C$ and $d(u,y) + |C| \le d(x,y') + |C| \le 3\sqrt{n}$ then according to \autoref{lemma:cycle-lemma}, $(s,t,e)$ belongs to Case 3.

For the rest of the proof we assume that $e \notin P_{T_x}(x,y')$ and prove that $(s,t,e)$ belongs to Case 4.
If $y' = t$ then $(s,t,e)$ belongs to Case 4 with $y=x$ as $e$ is $(s,x)$-replaceable (as $P(s,t,e)[s,x] = P_{T_s}(s,x)$ does not contain the edge $e$), $e$ is also $(x,t)$-replaceable (since $e \notin P_{T_x}(x,y') = P_{T_x}(x,t)$)  and $d(s,t,e) = d(s,x) + d(x,x) + d(x,t)$ (as $x$ is on $P(s,t,e)$).

If $y' \ne t$, then there exists a vertex $y \in B_2$ that hits the path $\Last_{T_x,  \sqrt{n}/2 }(y')$. 
Since $P(s,t,e)$ is a replacement path whose common prefix with $T_s$ is maximal, it must hold that $e \in  P_{T_s}(s,y)$ (otherwise there is a replacement path for $(s,t,e)$ whose prefix is $P_{T_s}(s,y)$ which is longer than the prefix $P_{T_s}(s,q)$ of $P(s,t,e)$). As $e \in  P_{T_s}(s,y)$, according to \autoref{lemma:edge-on-one-side-only} it holds that $e$ is $(y,t)$-replaceable. It follows that $(s,t,e)$ belongs to Case 4, as $e$ is $(s,x)$-replaceable (as $P(s,t,e)[s,x] = P_{T_s}(s,x)$ does not contain the edge $e$), $e$ is $(x,y)$-replaceable (as $e \notin  P_{T_x}(x,y')$, and as $P_{T_x}(x,y')$ contains $P_{T_x}(x,y)$ it follows that also $e \notin  P_{T_x}(x,y)$), $e$ is $(y,t)$-repalceable and $d(s,t,e) = d(s,x) + d(x,y) + d(y,t)$ (it is easy to see that $P_{T_s}(s,x) \circ P_{T_x}(x,y) \circ P_{T_y}(t,y)$ is a replacement path for $e$).
%
\end{proof}

\section{Derandomizing the Algorithm of Grandoni and Vassilevska Williams}
\label{sec:derand-grandoni}

\sloppy
We describe how to derandomize the algebraic SSRP algorithm of Grandoni and Vassilevska Williams~\cite{GrandoniVWilliamsFasterRPandDSO_journal} for undirected graphs with positive integer edge weights in the range $[1,M]$.
For a pair $(t, e) \in V \times E$, if $e$ does not lie along
$P_{T_s}(s,t)$, then $d(s,t,e) = d(s,t)$. 
The remaining pairs $(t, e)$ are called relevant, and we focus on them.

The first step in their algorithm is a partition of $T_s$ into a small (subpolynomial)
number of subtrees $T'$. Using balanced tree separators, they can guarantee that each $T'$
contains roughly the same number of nodes (modulo constants). 
Let $P'$ be the path from $s$ to the root of $T'$. 
For any relevant
pair $(t, e)$ there must exist some subtree $T'$ such that
$t \in V(T')$ and either 
(a) $e \in E(T')$ or 
(b) $e \in E(P')$. 

This way they identify a collection of subproblems, where
each subproblem is of the following two forms. In a subtree problem, we are given a subtree $T'$ of $T$ and we want to compute replacement paths $P(s,t,e)$ where both
$t$ and $e$ belong to $T'$ (handling (a) above).
In a subpath problem we are given a subpath $P'$ of $T$ from the source $s$ to a node $t'$,
and a subtree $T'$ of $T$ rooted at $t'$, 
and we want to compute replacement paths $P(s,t,e)$ with 
$t$ in $T'$ and $e$ in $P'$ (handling (b) above).
The subpath problems $(P', T')$ can be easily derandomized using previous results, we defer the description to the end of this section. 

They solve each subtree problem $T'$ recursively using randomization, after
a preliminary randomized compression step where they replace the nodes outside $T'$ 
with a subpolynomially smaller random subset $B_{rand}$ of them, 
adding auxilliary edges representing shortest paths between the sampled nodes.
We show how to derandomize the random selection of the set of pivots $B_{rand}$ by greedily computing a set of pivots $B_{greedy}$.

First, deterministically compute APSP in the graph $G-E(T')$ as in~\cite{ShZw99} (the shortest paths trees are computed deterministically using~\cite{AlNa96}).
Let $T'_v$ be the shortest paths tree rooted in $v$ in the graph $G-E(T')$.
Run the GreedyPivotsSelection algorithm to find in $\Otilde(n^2)$ time a hitting set $B_{greedy}$ of size $O(n \log n / H)$ that hits all the paths $\{ \Last_{T'_v, H}(u) \ | \ u, v \in V,  \ |\Last_{T'_v, H}(u)| = H \}$. 
Construct the complete graph $G'$ on node set $B_{greedy} \cup V(T') \cup \{s\}$ whose edges $e$ are labelled as follows.
For every $u,v \in B_{greedy} \cup V(T') \cup \{s\}$, if $d_{G-E(T')}(u,v) \le HM$ set $\omega'(u,v) = d_{G-E(T')}(u,v)$, otherwise set $\omega'(u,v) = \infty$. 
Then add back edges $e \in E(T')$ with their original weight $\omega'(e) = \omega(e)$.
To prove the correctness of our derandomization, we show that $G'$ contains a contracted
representative of each replacement path for the considered triples $(s, t, e)$.

\begin{lemma}
Let $t \in V(T')$ and $e \in E(T')$ be an edge that is on the path from $t'$ to $t$ in $T'$. Then $d_{G'}(s,t,e) = d_{G}(s,t,e)$. 
\end{lemma}

\begin{proof}
First we prove that $d_{G'}(s,t,e) \ge d_{G}(s,t,e)$. Observe that every edge in $G'$ that does not appear in $G$ represents a contraction of a path in $G - E(T')$. Thus, every path in $G'$ that avoids an edge $e \in E(T')$ is a contracted version of a path in $G$ that avoids the edge $e$, and hence $d_{G'}(s,t,e) \ge d_{G}(s,t,e)$.

Next, we prove that $d_{G'}(s,t,e) \le d_{G}(s,t,e)$.
We prove by induction that $d_{G'}(s,t,e) \le d_{G}(s,t,e)$ for every $t \in V(G')$.
For the base of the induction, it trivially holds that 
$d_{G'}(s,s,e) = d_{G}(s,s,e) = 0$.
For the inductive step, given a vertex $t \in V(G')$ the induction hypothesis is that for every $v \in V(G')$ such that 
$d_{G}(s,v,e) < d_{G}(s,t,e)$ it holds that
$d_{G'}(s,v,e) \le d_{G}(s,v,e)$, and we prove that
$d_{G'}(s,t,e) \le d_{G}(s,t,e)$. 
Let $P_G(s,t,e)$ be a replacement path for $(s,t,e)$ in $G$,
let $v \in V(G') \setminus \{t\}$ be the last vertex in $V(G')$ along $P_G(s,t,e)$. It follows that either $P_G(s,t,e)[v,t] = (v,t) \in E(T')$ or $P_G(s,t,e)[v,t]$ is a path in $G - E(T')$. If $P_G(s,t,e)[v,t] = (v,t) \in E(T')$ then $d_{G'}(s,t,e) \le d_{G'}(s,v,e) + \omega'(v,t) \le d_G(s,v,e) + \omega(v,t) = d(s,t,e)$, where the first inequality holds by the triangle inequality in $G'$ and the last inequality holds by the induction hypothesis and the fact that $\omega'(v,t) = \omega(v,t)$ as $(v,t) \in E(T')$.

We are left with the case that $P_G(s,t,e)[v,t]$ is a path in $G - E(T')$. Without loss of generality, we may assume that $P(s,t,e)$ is chosen such that $P_G(s,t,e)[v,t] = P_{T'_v}(v,t)$, and therefore $d_G(v,t,e) = d_{G - E(T')}(v,t)$. We claim that $|P_G(s,t,e)[v,t]| \le H$. Indeed, assume by contradiction that $|P_G(s,t,e)[v,t]| > H$, then by the greedy selection of $B_{greedy}$ it holds that at least one vertex $u \in B_{greedy}$ hits the path $\Last_{T'_v, H}(t)$, thus $u \in V(G')$ and also $|P_G(s,t,e)[u,t]| \le H < |P_G(s,t,e)[v,t]|$ which contradicts the assumption that $v$ is the last vertex of $V(G')$ along $P_G(s,t,e)$. 
Therefore, it holds that $|P_G(s,t,e)[v,t]| < H$ and hence $G'$ contains an edge $(v,t)$ with weight $\omega_{G'}(v,t) = d_G(v,t,e) \le MH$. 
We conclude that 
$d_{G'}(s,t,e) \le d_{G'}(s,v,e) + \omega_{G'}(v,t) \le  d_{G}(s,v,e) + d_{G}(v,t,e) = d_G(s,t,e)$ 
where the last inequality holds by the induction hypothesis and the fact that $\omega_{G'}(v,t) = d_{G}(v,t,e)$.
In summary, we have $d_{G'}(s,t,e) = d_{G}(s,t,e)$.
\end{proof}

Next we derandomize the subpath problem $(P', T')$, which is easy using previous results. Let $s$ and $t'$ are the endpoints of $P'$. 
In the worst case both $P'$ and $T'$ contain $O(n)$ nodes. 
Grandoni and Vassilevska Williams distinguish between two types of replacement paths 
$P(s,t,e)$ for $(t, e) \in V(T') \times E(P')$,
$e = uv$. A jumping path $P(s,t,e)$ leaves 
$P'$ at some node (between $s$ and $u$) and then meets $P'$ 
again at some other node (between $v$ and $t'$). 
A departing path $P(s,t,e)$ leaves $P'$ at some node (between $s$ and $u$) and never
meets $P'$ again.
It is easy to deal with jumping paths via a reduction to the Replacement Paths (RP) problem,
which is defined as SSRP but with a \emph{fixed} target $t$. 
We solve the RP problem for the $s$-$t'$-path $P'$ in time $\Otilde(Mn^\omega)$ 
with the deterministic algorithm of~Chechik and Nechushtan~\cite{ChNe20}.
Let $d'(s,t',e)$ be the resulting distances, then the
shortest jumping path length for the triple $(s, t, e)$ is
simply $d'(s,t',e) + d(t', t)$, taking $\Otilde(n^2)$ extra time prepare.
It remains to compute the departing paths. 
Observing that it is sufficient to compute all the
distances $d_{G'}(v, t)$ from nodes $v$ in $P'$ to nodes $t$ in 
$T'$ in the graph $G' := G - E(P')$. 
Let $s = v_1, v_2 \ldots v_h = t'$ be the sequence of nodes in $P'$. 
For $e = (v_i, v_{i+1})$ and any $t \in V(T')$, 
the shortest departing path for $(s, t, e)$ has
length $\min_{j \le i} \{d(s,v_j) + d_{G'}(v_j , t) \}$. 
For a fixed $t$, we can compute these quantities for all $e \in P'$ via a single
scan of the nodes of $P'$
from $v_1$ to $v_h$ (updating the
corresponding minimum each time). This takes $O(n^2)$ time. 
For the computation of the distances $d_{G'}(v, t)$
one can directly apply the deterministic APSP algorithm by Shoshan and Zwick~\cite{ShZw99}.
This solves the subpath problem with integer weights in $[1, M]$ in time
$\Otilde(Mn^\omega)$.

\section{Omitted Proofs of \autoref{sec:subquadratic_preprocessing}}
\label{app:subquadratic}

\sampling*

\begin{proof}
	We first prove that $|R| = \Otilde(n/L)$ with high probability.
	The expected size of $|R|$ is $\operatorname{E}[|R|] = n \nwspace \frac{c\ln n}{L}$.
	We use a Chernoff bound of the form
	$\operatorname{Pr}[ \nwspace |R| \ge (1+\delta) \operatorname{E}[|R|] \nwspace] 
		\le e^{-\delta^2 \operatorname{E}[|R|]/3}$
	for any $0 \le \delta \le 1$.
	Set $\delta = \sqrt{3c \ln n/\operatorname{E}[|R|]} = \sqrt{3L/n}$, which gives  
	$\operatorname{Pr}[\nwspace |R| \ge (1+\delta) \operatorname{E}[|R|] \nwspace] 
		\le e^{- c \ln n} = n^{-c}$.
	Finally, from
	$(1+\delta) \cdot \operatorname{E}[|R|] 
		= (1+ \sqrt{3c \ln n/\operatorname{E}[|R|]}) \cdot \operatorname{E}[|R|]
		< (3c \ln n) \cdot n \frac{c\ln n}{R} = \Otilde(n/L)$,
	we get that $|R| = \Otilde(n/L)$
	holds with probability of at least $1-n^{-c}$.
	
	For the second part, 
	observe that the probability of not sampling any of the vertices from $P \in \mathcal{P}$
	to be included in $R$ is at most
	$(1-\frac{c \nwspace\ln n}{L})^{|S_i|} \le (1-\frac{c \nwspace\ln n}{L})^{L}
		\le e^{-c \nwspace \ln n} = n^{-c}$.
	A union bound over the $\ell$ paths implies the claim.
\end{proof}

\unsuccessful*

\begin{proof}
	Fix $e \in [a,j]$ and let $P'$ be the shortest path from $s$ to $t$ in $G-e$
	that is forced to pass through some random pivot. 
	Let $\delta'=\min_{\chi \in R}\big\{d(s,\chi,e)+d(\chi,t)\big\}$ be the length of $P'$. 
	The value $\delta'$ is w.h.p.\ equal to the replacement distance $d(s,t,e)$.
	We assume this is the case.
	
	We say a vertex lies \emph{below} $e$ on $P(s,x_2)$
	if it is on the subpath that starts with $e$ and ends in $x_2$.
	We divide the proof into two cases, 
	depending on whether $P'$ remerges at a vertex below $e$ or not.
	If $P'$ does so, then it also runs through $x_2$ and $\delta' \geq d(s,x_2,e)+d(x_2,t)$.
	Therefore, $e$ is not in the far case II
	and thus is none of the distinguished edges $e^*_\ell$, $\ell \in [k]$.

	For the other case, recall that
	$\delta = \min_{\chi \in R}\big\{d(s,\chi,e_j)+d(\chi,t)\big\}$,
	$\mu = d(s,x_2,e_j)+d(x_2,t)$,
	and the fact that $e_j$ maximizes the last expression over the interval $[a,j]$.
	Now, if $P'$ does not run through any vertex of $P(s,x_2)$ that is below $e$, 
	then $P'$ also exists in $G-e_j$, whence $\delta' \geq \delta$.
	If the assumption $\delta \geq \mu$ is true,
	then $\delta' \geq d(s,x_2,e)+d(x_2,t)$ again follows
	and $e$ is not among $e^*_1,\dots,e^*_k$.	
	Finally, if the assumption $\delta > \Delta_{[a,b]}$ is true,
	then also $\delta'$ is larger than the upper bound $\Delta_{[a,b]}$,
	it is not admissible and $e \not \in \{e_1,\dots,e_k\}$.
\end{proof}

\successful*

\begin{proof}
	Since the search is successful, we have
	$\delta =\min_{\chi \in R}\big\{d(s,\chi,e_j)+d(\chi,t)\big\} < \mu = d(s,x_2,e_j)+d(x_2,t)$.
	We can be certain that edge $e_j$ is in the far case II with respect to target $t$.
	Let $\ell \in [k]$ be the (unique) index such that $R_\ell$
	is the representative replacement path for $e_j$.
	This means that $d_\ell = w(R_\ell) =  d(s,t,e_j)$
	and $\ell$ is indeed the maximum index for any distinguished edge in $[1,j]$.
	Moreover, with high probability $\delta$ equals the replacement distance $d(s,t,e_j)$,
	implying $d_{\ell} = \delta$.
	
	We argue next that $e^*_\ell$ is not in $[1, a{-}1]$.
	Let $e \in [1, a{-}1]$ be any edge in the far case II
	for which the replacement distance $d(s,t,e)$ is minimum.
	Then, we have $d(s,t,e) > \Delta_{[a,b]}$.
	Combining the successful search, that is, $\delta \le \Delta_{[a,b]}$,
	with $\delta_{\ell} = \delta$ (w.h.p.) shows that $e \neq e^*_{\ell}$.
	
	We are left to prove that indeed the equality $e_i = e^*_\ell$ holds.
	Recall that the random pivot $\chi_j $ is such that $\delta = d(s,\chi_j,e_j) + d(\chi_j,t)$
	and the replacement path $P(s,\chi_j,e_j)$ is chosen such
	that its divergence point is closest to $s$.
	Finally, the index $i \in [a,j]$ is defined to be the smallest one such that
	$P(s,\chi_j,e_j)$ also avoids $e_i$ and $\delta < d(s,x_2,e_i)+d(x_2,t)$.
	This implies that the edge $e_i$ also belongs to the far case II by the same argument as above.
	
	Since $i \le j$, we have $d(s,t,e_i) \ge d(s,t,e_j) = d_{\ell}$.
	Using the fact $\delta = w(P(s,\chi_j,e_j)) + d(\chi_j,t) \ge d(s,t,e_i)$
	yields $\delta = d(s,t,e_i) = d_{\ell}$ w.h.p.
	In other words, $e_i$ is the edge in the far~case~II closest to $s$
	with that exact replacement distance, which is the definition of $e^*_{\ell}$.
\end{proof}

\numberunsuccessfulsearches*

\begin{proof}
Let $\mathcal{T}$ be the recursion tree in which each node represents an interval $[a,b]$ 
and is labeled either successful or unsuccessful depending on the outcome of the search.
To each successful node $[a,b]$, we attach the information about the corresponding pair $(\delta_i,e_i)$
found while exploring that interval (not its upper and lower intervals).
Finally, the left and right children of each node correspond to
the lower and upper interval of the recursion, respectively. 

We first prove that any path in $\mathcal{T}$ from any node to any of its proper descendant 
that visits only nodes labeled as unsuccessful contains at most $O(\sqrt{Mn})$ nodes.
Let $\mu_1,\dots,\mu_{\ell}$ be the $\mu$-values computed in the consecutive unsuccessful searches,
in the order in which the intervals are considered by the algorithm.
Each $\mu_i$ refers to a certain path $P_i$ 
that diverges from $P(s,x_2)$ at a vertex $z_i$
and avoids a certain edge $f_i$, then remerges with $P(s,x_2)$, and finally ends in the target $t$.
This means, $\mu_i = w(P_i) = d(s,x_2,f_i) + d(x_2,t) \ge d(s,t,f_i)$.
By the definition of the algorithm, we have that $\mu_1 > \dots > \mu_\ell$.
Moreover, as an unsuccessful search in an interval causes a recursion only on its lower subinterval
(the one that is further away from the source $s$), 
we also have that $f_i$ is strictly closer to $s$ than $f_{i+1}$ and $z_{i+1}$. 
This in turn implies $d(s,z_i)< d(s,z_{i+1})$.

We define $\Offset_i=\mu_i-d(s,t)$.
Note that $\Offset_i \ge \ell - i$ as the values $\mu_i$ strictly decrease.
The detour part of $P_i$, starting at $z_i$,
has length strictly larger than $\Offset_i$.
(If $P_i$ remerges at vertex $y_i \neq z_i$, 
the precise length of the detour part is $\Offset_i + d(z_i,x_2) - d(y_i,x_2)$.)
The detour part thus has at least $\Offset_i/M$ vertices.
Let the \emph{stub} $S_i$ be formed by the first $\Offset_i/2M$ one of them.
No two stubs can intersect as this would give a shortcut to avoid edge $f_i$, for some $i$,
implying the contradiction $d(s,t,f_i) < \mu_i$.
(See also \autoref{lem:far_case_II_edge_failures}.)
In summary, this gives $n\geq \sum_{i=1}^{\ell} |V(S_i)| \geq \sum_{i=1}^{\ell} \frac{\ell-i}{2M}$ from which we derive $\ell = O(\sqrt{Mn})$.

If a target vertex $t$ has no replacement path that falls in the far case II ($k=0$),
the search tree contains only $O(\sqrt{Mn})$ unsuccessful searches and the lemma follows.
%
%
%
%
It remains to prove the case $k > 0$.
Any unsuccessful search occurs in an interval $[a,b]$ for which the lower bound $\delta_{[a,b]}$ is strictly positive and the value $\mu$ computed by the algorithm satisfies $\mu > \delta_{[a,b]}$. Using the same arguments as above, the maximum number of unsuccessful searches in intervals $[a_1,b_1],\dots,[a_\ell,b_\ell]$ with the same lower bound $\delta_{[a_1,b_1]}=\dots=\delta_{[a_\ell,b_\ell]}$ is $O(\sqrt{Mn})$.

Let $j \in [a,b]$ be the largest index such that $\mu = d(s,x_2,e_j) + d(x_2,t)$.
Any replacement path for an edge $e_h \in [j{+}1,b]$ in the lower subinterval
also diverges from $P(s,x_2)$ at a vertex in that subinterval. 
Otherwise, $P(s,x_2,e_h) \circ P(x_2,t)$,
which is strictly shorter than $P(s,x_2,e_j)\circ P(x_2,t)$,
would be a better path to avoid the edge $e_j$.
Thanks to this observation, a simple proof by induction shows that the path 
$P(s,x_2,e_j) \circ P(x_2,t)$ of length $\mu$ 
always diverges from $P(s,x_2)$ at a vertex in the interval $[a,j]$.
Recall that $\delta = \min_{\chi \in R} \{d(s,\chi,e_j) + d(\chi,t)\}$.
We (re-)define the stub $S_j$ to consist of the first $(\delta-d(s,t))/2M$ 
vertices of the detour part of $P(s,x_2,e_j)$ 
and we associate the stub with $[a,j]$. 
The main observation is that the stub $S_j$ cannot intersect with other stubs
that we define recursively on the lower interval $[j{+}1,b]$ 
as otherwise the two intersecting stubs would form a detour strictly shorter than $\delta-d(s,t)$. 
As this property is true for any interval, we have that all the stubs we defined are pairwise vertex-disjoint. 

Let $d_1< \dots< d_k$ be the lengths\footnote{%
	For notational convenience, we handle the lengths of the
	representative replacement paths here
	in the opposite order compared to \autoref{sec:subquadratic_preprocessing}.
}
of the $k = O(\sqrt{Mn})$ representative replacement paths for edges in the far case II 
with respect to target $t$.
For this part of the proof, we (re-)define $\Offset_i:=d_i-d(s,t)$. 
Let $n_i$ be the number of stubs that are associated to those intervals $[a,b]$ 
for which $\delta_{[a,b]}=d_i+1$. 
We have that the number of vertices spanned by the union of all $n_i$ such stubs is at least $n_i \cdot \Offset_i/2M \ge n_i \cdot i/2M$,
where the right-hand side
stems from sequence of $d_i$ being increasing.

The overall number of unsuccessful searches is therefore upper bounded by $\sum_{i=1}^{k}n_i$,
where the $n_i$ are subject to the following constraints.
First, we have $n_i \leq c\sqrt{Mn}$, or some constant $c > 0$
since the lower bound $\delta_{[a,b]}$ is the same for the $n_i$ stubs.
Furthermore, we get $\sum_{i=1}^k n_i \frac{i}{2M} \leq \sum_{i=1}^k n_i \frac{\Offset_i}{2M} \leq n$ 
from the stubs being pairwise disjoint.
The upper bound of the number of unsuccessful searches is maximized
if, for all indices $i \le i_{\max} =O(M^{1/4} \nwspace n^{1/4})$, we have the maximum $n_i=c \sqrt{Mn}$,
while for the remaining indices we have $n_i=0$.
This gives the estimate $\sum_{i=1}^k n_i \le c \sqrt{Mn} \cdot i_{\max} = O(M^{3/4} \nwspace n^{3/4})$.
\end{proof}

\subparagraph*{The near case.}

We now describe how to handle the replacement paths in the near case 
with respect to a fixed (regular) pivot $x\in D$. 
Using the same data structures that we presented in \Cref{sec:reduction_algorithm_to_DS,sec:subquadratic_preprocessing},
we can assume that each value $d(s,t,e)$ can be retrieved w.h.p.\ in $\Otilde(1)$ time
for each target vertex $t$ and edge $e$ on the path $P(s,D_2[t])$.

We denote by $V_{x}$ the set of target vertices $t$ such that $D_2[D_2[t]]=x$. Fix a failing edge $e=\{u,v\}$, with $u$ closer to $s$ than $v$, such that $d(x,v) \leq 4LM$. We construct the graph $G_e$ that contains $s$ plus the subset $V_e$ of vertices of $V_{x}$ that are below $e$ in $T_s$ and all the other vertices of $G-e$ that are connected with at least one vertex of $V_e$ by an edge. $G_e$ contains all the edges $\{u',v'\}$ of $G-e$ that are incident to some vertex of $V_e$. Moreover, for each vertex $t'$ of $G_e$ not in $V_e$ we add the edge $(s,t')$ of weight equal to $d(s,t',e)$. By our assumption, the value $d(s,t',e)$ is available w.h.p. since it refers to a replacement path that do not fall in the near case.

The shortest path tree of $G_e$ rooted in $s$ contains a compact representation of the replacement paths $P(s,t,e)$, for every $t \in V_e$. By compact we mean that $P(s,t,e)$ is the concatenation of the paths $P(s,t',e)$ and $P(t',t,e)$, where $t'$ is the vertex of $V\setminus V_e$ that precedes the first vertex of $V_e$ that is encountered while traversing the vertices of $P(s,t,e)$ in order from $s$ to $t$. By construction of the graph $G_e$, the replacement path $P(s,t',e)$ is modeled by the single edge $(s,t')$ of $G_e$. 

Let $m_x$ be the number of edges that are incident to the vertices of $V_x$. By construction, for each vertex $t \in V_x$ there are $O(ML)$ graphs $G_e$ such that $t \in V_e$. Therefore, the overall time to compute, for a fixed pivot $x$, the shortest path trees of all the corresponding graphs $G_e$, is $O(MLm_x)$. Since $t$ appears in the set $V_x$ for at most 4 distinct pivots -- i.e., $D_1[t], D_2[t],D_1[D_2[t]], D_2[D_2[t]]$ -- and as the sum of the values $m_x$, for all the pivots $x$, is at most twice the number of edges of the graph $G$, we have that the time complexity for computing all replacement paths that fall in the near case is $O(MLm)=O(M^{7/8} \nwspace m^{1/2} \nwspace n^{11/8})$.

\subparagraph{The Single-Source DSO.} 
We computed the relevant replacement distances as well as the vertices at which the corresponding
replacement paths diverge from the original shortest paths.
The Single-Source DSO can now be build using the same techniques we explained in \autoref{sec:reduction_algorithm_to_DS}.
Recall that we use a predecessor data structure to make the oracle path-reporting.
In particular, $\last(t,e)$ is the predecessor of $t$ in the shortest path tree of $G-e$ rooted at
the source vertex $s$, see the proof of \autoref{lem:path-reporting} in \autoref{app:omitted_reduction_proofs} for more details.

We denote by $D[t]$ the pivot of $t$ 
as defined in \autoref{sec:reduction_algorithm_to_DS}. We preprocess all the paths that have been computed in the near case defined w.r.t.\ pivot $x_2 = D_2[t]$ to check which of them fall in the far cases I and II defined w.r.t.\ pivot $D[t]$;
clearly, all others fall in the near case for $D[t]$.
This requires constant time per path if we visit, for each failing edge $e$, all the vertices in the set $V_e$ for graph $G_e$ and also allows us to keep track of $\last(t,e)$.  Once this classification has been done, we can build the oracle. 
Furthermore, we can also enable it to report the replacement paths.
In fact, for a fixed target node $t \neq s$, we need to store the edge incident to $t$ of each of the computed replacement paths that fall in the near case for pivot $D[t]$ as well as the edge incident to $t$ of each of the $O(\sqrt{Mn})$ computed replacement paths that fall in the far case II (again, for $D[t]$). For the replacement paths that fall in the far case I for $D[t]$, we already know that the edge entering $t$ is $\last(t)$, i.e., the predecessor of $t$ in $T_s$. This implies that we do not have to scan replacement paths that falls in the far~case~I for $x_2$
as such paths also fall in the far case I for $D[t]$.

\end{document}